\documentclass{article}

\usepackage{hyperref} 

\usepackage{booktabs} 
\usepackage[ruled]{algorithm2e} 

\SetAlFnt{\small}
\SetAlCapFnt{\small}
\SetAlCapNameFnt{\small}
\SetAlCapHSkip{0pt}
\IncMargin{-\parindent}

\usepackage{tabularx}

\usepackage[margin = 1in]{geometry}
\usepackage{bbding}
\usepackage[nointegrals]{wasysym}
\usepackage[utf8]{inputenc}
\usepackage{xspace}
\usepackage{mathtools}
\usepackage{url}
\usepackage{subcaption}
\usepackage{tikz}
\usetikzlibrary{arrows, fit, positioning}
\usetikzlibrary{backgrounds,automata,calc}
\usetikzlibrary{decorations.pathreplacing}
\usepackage{csquotes}
\usepackage{amsmath,amsfonts,amsthm}
\usepackage{bm,bbm}
\usepackage{pgfplots}

\DeclareMathOperator*{\argmax}{arg\,max}

\usepackage{enumerate}
\usepackage{paralist}
\usepackage[shortlabels]{enumitem}
\usepackage{appendix}
\usepackage[capitalize]{cleveref}
\usepackage{thmtools}
\usepackage{mathtools}
\usepackage{natbib}
\usepackage{changepage}

\usepackage{commath}

\newcommand{\OXS}{\text{-OXS}\xspace}
\newcommand{\ONSUB}{\text{-ONSUB}\xspace}
\newcommand{\SUB}{\text{-SUB}\xspace}
\newcommand{\ADD}{\text{-ADD}\xspace}
\newcommand{\USW}{\texttt{USW}\xspace}
\newcommand{\EFone}{\texttt{EF1}\xspace}
\newcommand{\PROP}{\texttt{PROP1}\xspace}
\newcommand{\MAXUSW}{\texttt{MAX-USW}\xspace}

\newcommand{\MMS}{\texttt{MMS}\xspace}

\newcommand{\R}{\mathbb{R}}

\renewcommand{\cal}[1]{\mathcal{#1}}

\newcommand{\lex}{\texttt{lex}}

\newtheorem{theorem}{Theorem}[section]
\newtheorem{lemma}[theorem]{Lemma}

\newtheorem{claim}[theorem]{Claim}
\newtheorem{prop}[theorem]{Proposition}

\newtheorem*{theorem*}{Theorem}

\DeclareMathOperator{\Welfare}{W}
\newcommand{\Uplambda}{\mathbin{\mathrm \Lambda}}

\usepackage{mathtools}
\DeclarePairedDelimiter{\ceil}{\lceil}{\rceil}

\setlist{topsep=0.5ex,itemsep=0.1ex}

\theoremstyle{definition}

\theoremstyle{definition}
\newtheorem{definition}[theorem]{Definition}
\newtheorem{obs}[theorem]{Observation}
\newtheorem{example}[theorem]{Example}

\title{The Good, the Bad and the Submodular: Fairly Allocating Mixed Manna Under Order-Neutral Submodular Preferences}
\author{Cyrus Cousins, Vignesh Viswanathan and Yair Zick \\
University of Massachusetts, Amherst \\
\texttt{\{cbcousins, vviswanathan, yzick\} @ umass.edu}}
\date{}

\begin{document}

\maketitle

\begin{abstract}
We study the problem of fairly allocating indivisible goods (positively valued items) and chores (negatively valued items) among agents with decreasing marginal utilities over items. 
Our focus is on instances where all the agents have {\em simple} preferences; specifically, we assume the marginal value of an item can be either $-1$, $0$ or some positive integer $c$. 
Under this assumption, we present an efficient algorithm to compute leximin allocations for a broad class of valuation functions we call {\em order-neutral} submodular valuations. 
Order-neutral submodular valuations strictly contain the well-studied class of additive valuations but are a strict subset of the class of submodular valuations.
We show that these leximin allocations are Lorenz dominating and approximately proportional. We also show that, under further restriction to additive valuations, these leximin allocations are approximately envy-free and guarantee each agent their maxmin share. We complement this algorithmic result with a lower bound showing that the problem of computing leximin allocations is NP-hard when $c$ is a rational number. 
 \end{abstract}

\section{Introduction}\label{sec:intro}
Fair allocation is a fundamental problem in computational economics. The problem asks how to divide a set of indivisible items among agents with subjective preferences (or valuations) over the items. 
Most of the literature focuses on the problem of dividing {\em goods} --- items with positive value. 
However, in several practical applications, such as dividing a set of tasks or allocating shifts to employees, items can be {\em chores} which provide a negative value to the agents they are allocated to. 

Fair allocation with mixed manna (instances containing both goods and chores) is, unsurprisingly, a harder problem than the case with only goods. 
Several questions which have been answered positively in the only goods setting are either still open or face a negative (impossibility or intractability) result in the mixed goods and chores setting. 
In particular, very little is known about the computability of \emph{leximin allocations}. A leximin allocation is one that maximizes the utility of the agent with the least utility; subject to that, it maximizes the second-least utility, and so on.

In the only goods setting, maximizing Nash welfare is arguably one of the most popular fairness objectives. 
Unfortunately, the Nash welfare of an allocation, defined as the product of agent utilities, loses its meaning in settings where agent utilities can be negative. In such settings, the leximin objective is a natural substitute. 
The leximin objective is easy to understand, and it implies an appealing egalitarian notion of fairness.
Therefore, designing algorithms that efficiently compute exact (or approximate) leximin allocations in the mixed manna setting is an important research problem. 

In this work, we present the first non-trivial results for this problem.

\subsection{Our Results}
We take a systematic approach towards solving this problem and start with a simple class of valuation functions. 
We assume all goods are symmetric and provide value $c$ (where $c$ is a positive integer), and all chores are symmetric and provide value $-1$. We also allow items to provide a value of $0$. 
We assume there are decreasing marginal gains over items; that is, after receiving many items, a good may provide value $0$ or even turn into a chore and provide a value of $-1$. 
We refer to this class of valuations as $\{-1, 0, c\}$-submodular valuations. 
With such valuations, there is no clear demarcation between goods and chores; an item may provide positive marginal value when added to an empty bundle but provide negative marginal value when added to a non-empty bundle.

We present an algorithm to compute a leximin allocation for a broad subclass of $\{-1, 0, c\}$ submodular valuations. 
More specifically, we show that leximin allocations can be computed efficiently when agents have $\{-1, 0, c\}$ submodular valuation functions that satisfy an additional property we call {\em order-neutrality}. 
Order-neutrality can be very loosely thought of as a property that requires the number of $c$-valued items in a bundle to be a monotonically non-decreasing function of the bundle.
We analyze these leximin allocations in further detail showing that they are Lorenz dominating and proportional up to one item. We also show that under further restriction to the case where agents have $\{-1, 0, c\}$ additive (or linear) valuations, leximin allocations are approximately envy-free as well as offer each agent their maxmin share.
We complement this result with lower bounds showing that the problem of computing leximin allocations becomes computationally intractable when $c$ is relaxed to being an arbitrary rational number as opposed to a positive integer.

\subsection{Our Techniques}\label{sec:techniques}
There are two main techniques we use in our algorithm design: {\em decomposition} and {\em path augmentation}.

\noindent\textbf{Decomposition:} We show (Lemma \ref{lem:3-decomposition}) that when agents have $\{-1, 0, c\}$ order-neutral submodular valuations, any allocation $X$ can be decomposed into three allocations $X^c$, $X^0$ and $X^{-1}$. Loosely speaking, $X^c$ contains all the items which provide a value of $c$, $X^0$ contains all the items which provide a value of $0$ and $X^{-1}$ contains all the items which provide a value of $-1$. 
Decompositions were previously used by \citet{cousins2023bivalued} in their study of bivalued submodular valuations, but showing decompositions exist is almost trivial when there are only two possible values an item can take. We use the order-neutrality property to show the existence of decompositions for a broader class of valuations. 
The existence of decompositions crucially allows us to separate the problem of computing a leximin allocation into three simple steps: we first build $X^c$ as a partial leximin allocation of all the $c$-valued items, then we construct $X^0$ to ensure $X^{-1}$ contains the fewest number of items possible. Finally, we build $X^{-1}$ by greedily giving chores to agents who place a high value on their allocated bundles. Indeed, showing that this simple approach works is highly non-trivial and relies on carefully using decomposition combined with the technique of path augmentation (Lemma \ref{lem:xc-equals-yc} and Theorem \ref{thm:leximin}). 

\noindent\textbf{Path Augmentation:} To build a partial leximin allocation consisting only of $c$-valued items, we use the technique of path augmentation. Path augmentation (mostly used with binary submodular valuations) uses multiple single-item transfers among agents to carefully manipulate allocations to make them leximin \citep{viswanathan2022yankee, viswanathan2022generalyankee}. The name `path augmentation' comes from the fact that these transfers are usually along a path --- an agent takes an item from another agent, who in turn takes an item from another agent, and so on. A direct application of path augmentation fails when agents have $\{-1, 0, c\}$ order-neutral submodular valuations. This is because updating $X^c$ using path augmentations can have undesirable effects on $X^0$ and $X^{-1}$. We show that by carefully choosing paths, we can avoid these undesirable effects. Specifically, we modify the classic path augmentation technique by giving a weight to each possible single item transfer. We show that choosing a minimum weight path (Theorems \ref{thm:Pareto-improving-paths} and \ref{thm:exchange-paths}) leads to a set of transfers that desirably modifies $X^c$ without undesirably affecting $X^0$ and $X^{-1}$.
  
\subsection{Related Work}\label{sec:related-work}
Fair allocation with mixed goods and chores has recently gained popularity in the literature. \citet{aziz2022mixedgoodsandchores} present definitions of envy-free up to one item (\EFone) and proportionality up to one item (\PROP) for the mixed goods and chores case; \citet{aziz2022mixedgoodsandchores} also present an algorithm to efficiently compute \EFone allocations. \citet{berczi2020envyfree} and \citet{bhaskar2021envyfreeness} further study the existence and computation of approximate envy-free allocations. There have also been a couple of papers studying maxmin share fairness with mixed goods and chores. \citet{Kulkarni2020ApproximatingMS} show that there exists a PTAS to compute a maxmin share fair allocation under certain assumptions. On the other hand, \citet{Kulkarni2021ptas} show that the problem of approximating the maxmin share of each agent is computationally intractable under general additive valuations.

To the best of our knowledge, there are only two papers \citep{berczi2020envyfree, bhaskar2021envyfreeness} on the fair allocation of mixed goods and chores which present results for non-additive valuation classes. \citet{berczi2020envyfree} present several positive results for very specific cases, such as identical valuations, Boolean valuations and settings with two agents. \citet{bhaskar2021envyfreeness} present an algorithm to compute \EFone allocations under {\em doubly monotone valuations}. Doubly monotone valuations assume each item is a good (always has positive marginal gain) or a chore (always has negative marginal gain) but other than this, do not place any restriction on the valuations. Order-neutral submodular valuations (discussed in this paper) relax the assumption that each item must be classified as a good or a chore, but come with the stronger restriction of submodularity. We also note that \citet{aziz2022mixedgoodsandchores} present an algorithm for computing an \EFone allocation for doubly monotone valuations but this result was disproved by \citet{bhaskar2021envyfreeness}.

The domain restrictions described above are not uncommon in fair allocation. In the mixed goods and chores setting, \citet{hosseini2023mixed} study the problem of fair allocation with lexicographic valuations --- a restricted subclass of additive valuations. In the goods setting, binary valuations \citep{halpern2020binaryadditive, viswanathan2022generalyankee, viswanathan2022yankee, Babaioff2021Dichotomous, benabbou2021MRF, barman2022groupstrategyproof} and bivalued valautions \citep{akrami2022halfintegers, akrami2022mnw, cousins2023bivalued, garg2021bivaluedadditive} have been extensively studied. In the chores setting, bivalued additive valuations \citep{ebadian2022bivaluedchores, aziz2023twotypes, garg2022bivaluedchores} and binary submodular valuations \citep{barman2023chores, viswanathan2023weighted} have been well studied. Our results in this paper are a natural extension of this line of work.

\section{Preliminaries}\label{sec:prelims}
We use $[k]$ to denote the set $\{1, 2, \dots, k\}$. Given a set $S$ and an element $o$, we use $S + o$ and $S - o$ to denote the sets $S \cup \{o\}$ and $S \setminus \{o\}$ respectively. 

We have a set of $n$ {\em agents} $N = [n]$ and a set of $m$ {\em items} $O = \{o_1, o_2, \dots, o_m\}$. Each agent $i \in N$ has a {\em valuation function} $v_i: 2^O \rightarrow \R$; $v_i(S)$ denotes the value of the set of items $S$ according to agent $i$. 
Given a valuation function $v$, we let $\Delta_v(S,o) = v(S+o) - v(S)$ denote the marginal utility of adding the item $o$ to the bundle $S$ under $v$. When clear from context, we sometimes write $\Delta_i(S,o)$ instead of $\Delta_{v_i}(S,o)$ to denote the marginal utility of giving the item $o$ to agent $i$ given that they have already been assigned the bundle $S$.

An {\em allocation} $X = (X_0, X_1, \dots, X_n)$ is an $(n+1)$-partition of the set of items $O$. $X_i$ denotes the set of items allocated to agent $i$ and $X_0$ denotes the set of unallocated items. Our goal is to compute {\em complete} fair allocations --- allocations where $X_0 = \emptyset$. When we construct an allocation, we sometimes only define the allocation to each agent $i \in N$; the bundle $X_0$ is implicitly assumed to contain all the unallocated items. Given an allocation $X$, we refer to $v_i(X_i)$ as the utility of agent $i$ under the allocation $X$. We also define the utility vector of an allocation $X$ as the vector $\vec u^X = (v_1(X_1), v_2(X_2), \dots, v_n(X_n))$.

We define two common methods to compare vectors. We will use these methods extensively in our analysis when comparing allocations. 

\begin{definition}[Lexicographic Dominance]
   A vector $\vec y \in \R^n$ {\em lexicographically dominates} a vector $\vec z \in \R^n$ (written $\vec y \succ_{\lex} \vec z$) if there exists a $k \in [n]$ such that for all $j \in [k-1]$, $y_j = z_j$ and $y_k > z_k$. We sometimes say an allocation $X$ lexicographically dominates an allocation $Y$ if $\vec u^X \succ_{\lex} \vec u^Y$. 
\end{definition}

\begin{definition}[Pareto Dominance]
A vector $\vec y \in \R^n$ {\em Pareto dominates} a vector $\vec z \in \R^n$ if for all $j \in [n]$, $y_j \ge z_j$ with the inequality being strict for at least one $j \in [n]$.
An allocation $X$ Pareto dominates an allocation $Y$ if $\vec u^X$ Pareto dominates $\vec u^Y$.
\end{definition}

Note that Pareto domination implies lexicographic domination.

\subsection{Valuation Functions}
In this paper, we will be mainly dealing with the popular class of submodular valuations.
\begin{definition}[Submodular]\label{def:submodular}
A function $v: 2^O \rightarrow \R$ is a {\em submodular function} if
\begin{inparaenum}[(a)]
    \item $v(\emptyset) = 0$, and 
    \item for any $S \subseteq T \subseteq O$ and $o \in O\setminus T$, $\Delta_v(S, o) \ge \Delta_v(T, o)$.
\end{inparaenum}
\end{definition}

We also define restricted submodular valuations to formally capture instances where the set of possible marginal values is limited. Throughout this paper, we use the set $A$ to denote an arbitrary set of real numbers.
\begin{definition}[$A$\SUB functions]\label{def:a-sub}
Given a set of real numbers $A$, a function $v: 2^O \rightarrow \R$ is an $A$\SUB function if it is submodular, and every item's marginal contribution is in $A$. That is, for any set $S \subseteq O$ and an item $o \in O \setminus S$, $\Delta_v(S, o) \in A$.
\end{definition}

Our analysis of submodular functions requires that they satisfy an additional property we call {\em order-neutrality}. 
Given a submodular valuation function $v: 2^O \rightarrow \R$, the value of a set of items $S$ can be computed by adding items from $S$ one by one into an empty set and adding up the $|S|$ marginal gains. 
More formally, given a bijective mapping $\pi: [|S|] \rightarrow S$ which defines the order in which items are added, $v(S)$ can be written as the following telescoping sum:
\begin{align*}
    v(S) = \sum_{j \in [|S|]} \Delta_v \left (\mathsmaller{\bigcup}_{\ell \in [j-1]} \pi(\ell), \pi(j) \right )
\end{align*}
While the value $v(S)$ does not depend on $\pi$, the values of the marginal gains in the telescoping sum may depend on $\pi$. Given a set $S$ and a bijective mapping $\pi: [|S|] \rightarrow S$ we define the vector $\vec {v}(S, \pi)$ as the vector of marginal gains (given below) {\em sorted in ascending order}
\begin{align*}
    \bigg (\Delta_v \big (\emptyset, \pi(1) \big ), \Delta_v \big (\pi(1), \pi(2) \big ), \dots, \Delta_v \left (\mathsmaller{\bigcup}_{\ell \in [j-1]} \pi(\ell), \pi(j) \right ), \dots, \Delta_v \big (S - \pi(|S|), \pi(|S|) \big ) \bigg )
\end{align*}

We refer to this vector $\vec v(S, \pi)$ as a {\em sorted telescoping sum vector} . For any ordering $\pi$ and set $S$, the sum of the elements of $\vec v(S, \pi)$ is equal to $v(S)$. A submodular function $v$ is said to be {\em order-neutral} if for all bundles $S \subseteq O$ and any two orderings $\pi, \pi': [|S|] \rightarrow S$ of items in the bundle $S$, we have $\vec{v}(S, \pi) = \vec{v}(S, \pi')$; that is, any sorted telescoping sum vector is independent of the order $\pi$. 
For order-neutral submodular valuations, we sometimes drop the $\pi$ and refer to any sorted telescoping sum vector using $\vec v(S)$. 
This definition can be similarly extended to $A$\SUB functions. For readability, we refer to order-neutral $A$\SUB functions as $A$\ONSUB functions.

\begin{definition}[$A$\ONSUB functions]\label{def:a-onsub}
    Given a set of real numbers $A$, a function $v:2^O \rightarrow \R$ is an $A$\ONSUB function if it is both order-neutral and an $A$\SUB function.
\end{definition}

We will focus on instances with $\{-1, 0, c\}\ONSUB$ valuations where $c$ is a positive integer. 
Throughout this paper, the only use of $c$ will be to denote an arbitrary positive integer.
We also present some results for the restricted setting where agents have additive valuations. 
\begin{definition}[$A\ADD$ functions]\label{def:a-add}
    Given a set of real numbers $A$, a function $v:2^O \rightarrow \R$ is an $A$-additive (or simply $A\ADD$) function if $v(\{o\}) \in A$ for all $o \in O$ and $v(S) = \sum_{o \in S} v(\{o\})$ for all $S \subseteq O$.
\end{definition}
To build intuition for the class of $\{-1, 0, c\}\ONSUB$ valuations, we present a few simple examples below. 

\begin{example}
Let $O = \{o_1, o_2, o_3, o_4\}$, the following functions $v_1, v_2, v_3: 2^O \rightarrow \R$ are $\{-1,0, c\}\ONSUB$:
\begin{align*}
    &v_1(S) = c\min\{|S|, 2\} \\
    &v_2(S) = c\mathbb{I}\{o_1 \in S\} - \mathbb{I}\{o_2 \in S\} \\
    &v_3(S) = c\min\{|S \cap \{o_1, o_2\}|, 1\} - |S \cap \{o_3, o_4\}|
\end{align*}
The valuation $v_1$ describes a function where any item provides a value of $c$ but the marginal utility of any item drops to $0$ after $2$ items are added to the bundle. $v_2$ describes a simple additive function where $o_1$ provides a value of $c$ and $o_2$ provides a value of $-1$. $v_3$ describes a slightly more complex function where $o_1$ and $o_2$ are goods but at most one of them can a provide a marginal value of $c$; $o_3$ and $o_4$ are chores and provide a marginal value of $-1$ each.
\end{example}

Note that additive valuations are trivially order-neutral submodular valuations. 
Unsurprisingly, not all submodular valuations are order-neutral --- consider a function $v$ over two items $\{o_1, o_2\}$ such that $v(\{o_1\}) = 0$, $v(\{o_2\}) = 1$ and $v(\{o_1, o_2\}) = 0$. This function is submodular, but not order-neutral, since $v(\{o_1, o_2\})$ has two different sorted telescoping sum vectors.
However, it is worth noting that there are many interesting non-additive order-neutral submodular functions. For example, as we will formally show later, any binary submodular function ($\{0, 1\}\SUB$) is order-neutral (\Cref{prop:a2-orderneutral}). We present a broad discussion about order-neutral submodular valuations in Section \ref{sec:discussion}.
\subsection{Fairness Objectives}\label{sec:fairness-objectives}
There are several reasonable fairness objectives used in the fair allocation literature. We discuss most of them in this paper. However, to avoid an overload of definitions, we only define the following two fairness objectives in this section. 

\noindent\textbf{Utilitarian Social Welfare (\USW)}: The utilitarian social welfare of an allocation $X$ is $\sum_{i \in N} v_i(X_i)$. An allocation $X$ is said to be \MAXUSW if it maximizes the utilitarian social welfare.

\noindent\textbf{Leximin: } An allocation is said to be leximin if it maximizes the utility of the least valued agent, and subject to that, the utility of the second-least valued agent, and so on \citep{viswanathan2022yankee, Kurokawa2018dichotomousallocation}. This is usually formalized using the sorted utility vector. The {\em sorted utility vector} of an allocation $X$ (denoted by $\vec s^X$) is defined as the utility vector $\vec u^X$ sorted in ascending order. An allocation $X$ is leximin if there is no allocation $Y$ such that $\vec s^Y \succ_{\lex} \vec s^X$.


Other objectives like maxmin share, proportionality and envy-freeness are defined in Section \ref{sec:leximin-properties}.

\subsection{Exchange Graphs and Path Augmentations} 
In this section, we describe the classic technique of path augmentations. A modified version of these path augmentations is used extensively in our algorithm design. 
Path augmentations have been used to carefully manipulate allocations when agents have binary submodular valuations ($\{0, 1\}\SUB$ functions). However, they only work with {\em clean} allocations.

\begin{definition}[Clean Allocation]\label{def:clean}
For any agent $i \in N$, a bundle $S$ is {\em clean} with respect to the binary submodualar valuation $\beta_i$ if $\beta_i(S) = |S|$. An allocation $X$ is said to be clean (w.r.t. $\{\beta_h\}_{h \in N}$) if for all agents $i \in N$, $\beta_i(X_i) = |X_i|$. 
\end{definition}

When agents have binary submodular valuations $\{\beta_h\}_{h \in N}$, given a clean allocation $X$ (w.r.t. $\{\beta_h\}_{h \in N}$), we define the {\em exchange graph} $\cal G(X, \beta)$ as a directed graph over the set of items $O$, where an edge exists from $o$ to $o'$ in the exchange graph if $o \in X_j$ and $\beta_j(X_j - o+o') = \beta_j(X_j)$ for some $j \in N$. There is no outgoing edge from any item in $X_0$.

Let $P = (o_1, o_2, \dots, o_t)$ be a path in the exchange graph $\cal G(X, \beta)$ for a clean allocation $X$. 
We define a transfer of items along the path $P$ in the allocation $X$ as the operation where $o_t$ is given to the agent who has $o_{t-1}$, $o_{t-1}$ is given to the agent who has $o_{t-2}$, and so on until finally $o_1$ is discarded and becomes freely available. 
This transfer is called {\em path augmentation}; the bundle $X_i$ after path augmentation with the path $P$ is denoted by $X_i \Uplambda P$ and defined as $X_i \Uplambda P = (X_i - o_t) \oplus \{o_j, o_{j+1} : o_j \in X_i\}$, where $\oplus$ denotes the symmetric set difference operation. 

For any clean allocation $X$ and agent $i$, we define $F_{\beta_i}(X, i) = \{o \in O: \Delta_{\beta_i}(X_i, o) = 1\}$ as the set of items which give agent $i$ a marginal gain of $1$ under the valuation $\beta_i$. 
For any agent $i$, let $P = (o_1, \dots, o_t)$ be a {\em shortest} path from $F_{\beta_i}(X, i)$ to $X_j$ for some $j \ne i$. 
Then path augmentation with the path $P$ and giving $o_1$ to $i$ results in a clean allocation where the size of $i$'s bundle $|X_i|$ goes up by $1$, the size of $j$'s bundle goes down by $1$ and all the other agents do not see any change in size. 
This is formalized below and exists in \citet[Lemma 1]{Barman2021MRFMaxmin} and \citet[Lemma 3.7]{viswanathan2022yankee}.

\begin{lemma}[\citet{Barman2021MRFMaxmin}, \citet{viswanathan2022yankee}]\label{lem:path-augmentation}
    Let $X$ be a clean allocation with respect to the binary submodular valuations $\{\beta_h\}_{h \in N}$. Let $P = (o_1, \dots, o_t)$ be a shortest path in the exchange graph $\cal G(X, \beta)$ from $F_{\beta_i}(X, i)$ to $X_j$ for some $i \in N$ and $j \in N + 0 - i$. Then, the following allocation $Y$ is clean with respect to $\{\beta_h\}_{h \in N}$.
    \begin{align*}
        Y_k = 
        \begin{cases}
            X_k \Uplambda P & (k \in N + 0 - i) \\
            X_i \Uplambda P + o_1 & (k = i)
        \end{cases}
    \end{align*}
    Moreover, for all $k \in N + 0 - i - j$, $|Y_k| = |X_k|$, $|Y_i| = |X_i| + 1$ and $|Y_j| = |X_j| - 1$. 
\end{lemma}

We also present sufficient conditions for a path to exist. 
A slight variant of the following lemma appears in \citet[Theorem 3.8]{viswanathan2022yankee}.

\begin{lemma}[\citet{viswanathan2022yankee}]\label{lem:augmentation-sufficient}
Let $X$ and $Y$ be two clean allocations with respect to the binary submodular valuations $\{\beta_h\}_{h \in N}$. For any agent $i \in N$ such that $|X_i| < |Y_i|$, there is a path from $F_{\beta_i}(X, i)$ to either
\begin{enumerate}[(i)]
    \item some item in $X_k$ for some $k \in N$ in the exchange graph $\cal G(X, \beta)$ such that $|X_k| > |Y_k|$, or
    \item some item in $X_0$ in $\cal G(X, \beta)$. 
\end{enumerate}
\end{lemma}

This technique of path augmentations has been extensively exploited in the design of the Yankee Swap algorithm \citep{viswanathan2022yankee}. Given an instance with binary submodular valuations, Yankee Swap computes a clean \MAXUSW leximin allocation in polynomial time. We use this procedure as a subroutine in our algorithm to compute a partial allocation.

\begin{theorem}[Yankee Swap \citep{viswanathan2022yankee}]\label{thm:yankee-swap}
When agents have binary submodular valuations $\{\beta_h\}_{h \in N}$, there exists an efficient algorithm that can compute a clean \MAXUSW leximin allocation.
\end{theorem}

We note that Yankee Swap is not the only algorithm to compute clean leximin allocations; \citet{Babaioff2021Dichotomous} also present an efficient algorithm to do so. 

\section{Understanding $A$\ONSUB Valuations}\label{sec:understanding}
In this section, we present some important results about $A$\ONSUB valuations exploring its connections to binary submodular functions.
These connections allow us to adapt path augmentations for our setting.
Our arguments, in this section, are a generalization of the arguments presented by \citet[Section 3]{cousins2023bivalued} about bivalued submodular valuations.

Our main result shows that given a threshold value $\tau \in \R$ and an $A$\ONSUB function $v_i$, the number of values in the sorted telescoping sum vector greater than or equal to the threshold $\tau$ corresponds to a binary submodular function. More formally, for any bundle $S \subseteq O$, let $\beta^{\tau}_i(S)$ denote the number of values in the sorted telescoping sum vector $\vec v_i(S)$ greater than or equal to $\tau$. We show that the function $\beta^{\tau}_i$ is a binary submodular function.

\begin{lemma}\label{lem:onsub-mrf}
For any $i \in N$ and $\tau \in \R$, $\beta^{\tau}_i$ is a binary submodular function.
\end{lemma}
\begin{proof}
We need to show three properties about $\beta^{\tau}_i$: 
\begin{inparaenum}[(a)]
    \item $\beta^{\tau}_i(\emptyset) = 0$,
    \item for any set $S$ and item $o \in O \setminus S$, $\Delta_{\beta^{\tau}_i}(S, o) \in \{0, 1\}$, and 
    \item for any sets $S \subseteq T \subseteq O$, and item $o \in O \setminus T$, $\Delta_{\beta^{\tau}_i}(S, o) \ge \Delta_{\beta^{\tau}_i}(T, o)$.
\end{inparaenum}

\noindent\textbf{(a) $\beta^{\tau}_i(\emptyset) = 0$:} This is trivial. The size of the sorted telescoping sum vector is $0$ and therefore, the number of elements in the sorted telescoping sum vector greater than or equal to $\tau$ is $0$ as well. 

\noindent\textbf{(b) for any set $S$ and item $o \in O \setminus S$, $\Delta_{\beta^{\tau}_i}(S, o) \in \{0, 1\}$:} We construct an order $\pi$ to compute the sorted telescoping sum vector for $S+o$ as follows: we add the items in $S$ first in an arbitrary order and then add $o$. If $\Delta_{v_i}(S, o) \ge \tau$, there will be $\beta^{\tau}_i(S) + 1$ elements in the vector $\vec{v_i}(S + o, \pi)$ greater than or equal to $\tau$. Otherwise there will only be $\beta^{\tau}_i(S)$ elements greater than or equal to $\tau$. Therefore, $\Delta_{\beta^{\tau}_i}(S, o) \in \{0, 1\}$.

\noindent\textbf{(c) for any sets $S \subseteq T \subseteq O$, and item $o \in O \setminus T$, $\Delta_{\beta^{\tau}_i}(S, o) \ge \Delta_{\beta^{\tau}_i}(T, o)$:} 
Using (b), we only need to show that if $\Delta_{\beta^{\tau}_i}(T, o) = 1$, then $\Delta_{\beta^{\tau}_i}(S, o) = 1$ as well. Assume $\Delta_{\beta^{\tau}_i}(T, o) = 1$. Construct the order $\pi$ to compute the sorted telescoping sum vector for $T+o$ as follows: we add the items in $T$ first in any arbitrary order and then add $o$. From $\Delta_{\beta^{\tau}_i}(T, o) = 1$, we can infer that $\Delta_{v_i}(T, o) \ge \tau$. Using submodularity, $\Delta_{v_i}(S, o) \ge \tau$ as well. Therefore, $\Delta_{\beta^{\tau}_i}(S, o) = 1$.
\end{proof}

This lemma is particularly useful since, if any allocation is clean with respect to the valuations $\{\beta^{\tau}_h\}_{h \in N}$, we can use path augmentations to update the allocation. In our analysis, we will extensively use the valuations $\{\beta^{0}_h\}_{h \in N}$ and $\{\beta^{c}_h\}_{h \in N}$ corresponding to the cases where $\tau = 0$ and $\tau = c$ respectively.
This lemma also allows us to decompose any allocation into two different allocations, one that consists of the items valued greater than or equal to the threshold and one that consists of the remaining items. More formally, given a set of $A$\ONSUB valuations $\{v_i\}_{i \in N}$ and a threshold $\tau$, we define corresponding binary submodular valuations for each agent $\{\beta^{\tau}_i\}_{i \in N}$. 
We can decompose any allocation $X$ into two allocations $X^{\ge \tau}$ and $X^{< \tau}$ --- the former containing all the items with value at least the threshold and the latter containing all the other items. Keeping with the convention set by \citet{cousins2023bivalued}, we refer to $X^{\ge \tau}$ as the {\em clean allocation} and $X^{< \tau}$ as the {\em supplementary allocation}. This is consistent with the definition of a clean allocation in Definition \ref{def:clean}.
We have the following result. 
\begin{lemma}\label{lem:decomposition}
    Given an instance with $A$\ONSUB valuations, for any allocation $X$ and threshold $\tau$, there exists clean and supplementary allocations $X^{\ge \tau}$ and $X^{< \tau}$ such that for each agent $i \in N$: 
    \begin{inparaenum}[(a)]
        \item $X^{\ge \tau}_i \cup X^{< \tau}_i = X_i$, 
        \item $X^{\ge \tau}_i \cap X^{< \tau}_i = \emptyset$, and 
        \item $\beta^{\tau}_i(X_i) = \beta^{\tau}_i(X^{\ge \tau}_i) = |X^{\ge \tau}_i|$.
    \end{inparaenum}
\end{lemma}
\begin{proof}
For each agent $i \in N$, pick any arbitrary ordering $\pi: [|X_i|] \rightarrow X_i$ over the bundle $X_i$ and compute the telescoping sum vector 
\begin{align*}
    \bigg (\Delta_{v_i} \big (\emptyset, \pi(1) \big ), \Delta_{v_i} \big (\pi(1), \pi(2) \big ), \dots, \Delta_{v_i} \big (\bigcup_{\ell \in [j-1]} \pi(\ell), \pi(j) \big ), \dots, \Delta_{v_i} \big (S - \pi(|X_i|), \pi(|X_i|) \big ) \bigg ).
\end{align*}
We know that there are $\beta^{\tau}(X_i)$ elements in the above vector which are greater than or equal to $\tau$. Each element is associated with the marginal gain of adding an item $o$; we construct $X^{\ge \tau}_i$ as the set of these $\beta^{\tau}(X_i)$ items. We define $X^{<\tau}_i$ as the set $X_i \setminus X^{\ge \tau}_i$. 

We do this for each $i \in N$ to construct the allocations $X^{\ge \tau}$ and $X^{< \tau}$. It is easy to see that this construction satisfies the properties (a)--(c).
\end{proof} 

We refer to the splitting of $X$ into $X^{\ge \tau}$ and $X^{< \tau}$ as a {\em decomposition}. 
Given an allocation $X$, we refer to such a decomposition using the shorthand notation $X = X^{\ge \tau} \cup X^{< \tau}$. More generally, given any two allocations $X$ and $Y$, we refer to the allocation $X \cup Y$ as the allocation where each agent $i \in N$ receives the bundle $X_i \cup Y_i$. 
In our algorithm, we decompose allocations using the threshold $c$ and then use path augmention to modify $X^{\ge c}$.

We also show that for any $A\ONSUB$ function $v$ and threshold $\tau$, the value of any bundle w.r.t. $\beta^{\tau}$ can be computed using $O(m)$ queries to $v$.
\begin{lemma}\label{lem:mrf-oracle}
Let $v$ be an $A\ONSUB$ function and $\tau$ be a real number. For any bundle $S \subseteq O$, $\beta^{\tau}(S)$ can be computed using $O(m)$ queries to $v$.
\end{lemma}
\begin{proof}
We use the following simple algorithm: we pick any arbitrary ordering $\pi$ and compute the telescoping sum vector $\vec v(S)$. $\beta^{\tau}(S)$ corresponds to the number of elements with value at least $\tau$ in the telescoping sum vector. Since the telescoping sum vector has size at most $m$, the lemma follows.
\end{proof}

\section{Understanding $\{-1, 0, c\}$\ONSUB Valuations}
We turn our attention to the specific set of valuations assumed in this paper --- $\{-1, 0, c\}\ONSUB$ valuations for some positive integer $c$. We establish a few important technical lemmas for fair allocation instances when all agents have $\{-1, 0, c\}\ONSUB$ valuations.

We first show that any allocation $X$ can be decomposed into {\em three} allocations $X^{-1}$, $X^0$ and $X^c$ such that the items with marginal value $c$ are in $X^c$, items with marginal value $0$ are in $X^0$ and items with marginal value $-1$ are in $X^{-1}$.

\begin{lemma}\label{lem:3-decomposition}
When agents have $\{-1, 0, c\}\ONSUB$ valuations, for any allocation $X$, there exist three allocations $X^{-1}$, $X^0$, and $X^c$ such that for each agent $i \in N$
\begin{inparaenum}[(a)]
    \item $X^{-1}_i \cup X^0_i \cup X^c_i = X_i$, 
    \item $X^{-1}_i$, $X^0_i$, and $X^c_i$ are pairwise disjoint,
    \item $v_i(X^c_i \cup X^0_i) = v_i(X^c_i) = c|X^c_i|$, and
    \item $v_i(X_i) = c|X^c_i| - |X^{-1}_i|$.
\end{inparaenum}
\end{lemma}
\begin{proof}
 We decompose $X$ with respect to the valuations $\{\beta^{0}_h\}_{h \in N}$ to create $X^{<0}$ and $X^{\ge 0}$ (Lemma \ref{lem:decomposition}). We then decompose $X^{\ge 0}$ with respect to the valuations $\{\beta^{c}_h\}_{h \in N}$ to create $X^{\ge 0, < c}$ and  $X^{\ge 0, \ge c}$ where $X^{\ge 0, \ge c}$ is the clean allocation and $X^{\ge 0, < c}$ is the supplementary allocation.

 We denote the allocations $X^{<0}$, $X^{\ge 0, < c}$ and  $X^{\ge 0, \ge c}$ as the allocations $X^{-1}, X^0$ and $X^c$ respectively. This construction trivially satisfies (a) and (b). To show that it satisfies (c) and (d), we fix an agent $i \in N$ and examine the sorted telescoping sum vector $\vec v_i(X_i)$. 

The sorted telescoping sum vector $\vec v_i(X_i)$ has up to three unique values --- $-1$, $0$ and $c$. By definition (Lemma \ref{lem:decomposition}), $|X^{-1}_i|$ (or $|X^{<0}_i|$) consists of the number of indices with value $-1$ in $\vec v_i(X_i)$ and  $|X^{\ge 0}_i|$ consists of the number of indices with value either $0$ or $c$. Similarly, since $\vec v_i(X^{\ge 0}_i)$ only consists the values $0$ and $c$, $|X^c_i|$ (or $|X^{\ge 0, \ge c}_i|$) corresponds to the number of indices with value $c$ in $\vec v_i(X^{\ge 0}_i)$. Since all the other indices in $\vec v_i(X^{\ge 0}_i)$ have value $0$, we have $v_i(X^0_i \cup X^c_i) = v_i(X^{\ge 0}_i) = c|X^c_i|$. Furthermore, by the definition of a decomposition, $v_i(X^c_i) = c|X^c_i|$. This proves (c).

To prove (d), we use the observation that the number of indices with value $c$ in $\vec v_i(X^{\ge 0}_i)$ is equal to the number of indices with value $c$ in $\vec v_i(X_i)$. $|X^{-1}_i|$ denotes the number of indices with value $-1$ in $\vec v_i(X_i)$ and all the indices which neither have the value $-1$ nor $c$ have the value $0$. Therefore, since $v_i(X_i)$ corresponds to the sum of the elements in $\vec v_i(X_i)$, we can conclude that $v_i(X_i) = c|X^c_i| - |X^{-1}_i|$.
\end{proof}

Similar to Lemma \ref{lem:decomposition}, we use $Y = Y^c \cup Y^0 \cup Y^{-1}$ to denote the decomposition of any allocation $Y$ into three allocations satisfying the conditions of Lemma \ref{lem:3-decomposition}.
We present an example of a decomposition below.

\begin{example}\label{ex:3-decomposition}
Consider an instance with two agents $\{1, 2\}$ and four items $\{o_1, o_2, o_3, o_4\}$. Agent valuations are defined as follows:
\begin{align*}
    v_1(S) = c\min\{|S \cap \{o_1, o_2,\}|, 1\}, &&
    v_2(S) = - |S \cap \{o_3, o_4\}|.
\end{align*}
Consider an allocation $X$ where $X_1 = \{o_1, o_2\}$ and $X_2 = \{o_3, o_4\}$. The following is a valid decomposition of $X$:
\begin{align*}
    X^c_0 = \{o_2, o_3, o_4\} && X^0_0 = \{o_1, o_3, o_4\} && X^{-1}_0 = \{o_1, o_2\}\\
    X^c_1 = \{o_1\} && X^0_1 = \{o_2\} && X^{-1}_1 = \emptyset \\
    X^c_2 = \emptyset && X^0_2 = \emptyset && X^{-1}_1 = \{o_3, o_4\}
\end{align*}
There may be other decompositions as well.
Specifically, if we swap $o_1$ and $o_2$ in the above decomposition, the new set of allocations still corresponds to a valid decomposition.
\end{example}

Note that while order-neutral submodular valuations can have multiple decompositions, $\{-1, 0, c\}\ADD$ valuations have a unique decomposition --- for any allocation $X = X^c \cup X^0 \cup X^{-1}$ and agent $i$, $X^c_i$ consists of all items in $X_i$ that agent $i$ values at $c$, $X^0_i$ consists of the items that $i$ values at $0$ and $X^{-1}_i$ consists of the items that $i$ values at $-1$.

\subsection{Weighted Exchange Graphs}\label{sec:weighted-exchange-graph}
Given an allocation $X = X^c \cup X^0 \cup X^{-1}$, the path augmentation technique introduced in Section \ref{sec:prelims} can be used to manipulate $X^c$ (using the exchange graph $\cal G(X^c, \beta^c)$) or $X^c \cup X^0$ (using the exchange graph $\cal G(X^c \cup X^0, \beta^0)$). However, when we use path augmentation with the exchange graph $\cal G(X^c, \beta^c)$ to manipulate the allocation $X^c$, we may affect the cleanness of $X^c \cup X^0$ (w.r.t. the valuations $\{\beta^0_h\}_{h \in N}$). To see why, consider the following simple example.

\begin{example}\label{ex:classic-paths-do-not-work}
Consider an example with one agent $\{1\}$ and two items $\{o_1, o_2\}$. The agent's valuation function is defined as follows:
\begin{align*}
    v_1(S) = c\min\{|S|, 1\} - \max\{|S|-1, 0\}.
\end{align*}
In simple words, the first item in the bundle gives agent $1$ a value of $c$ but the second item gives agent $1$ a marginal value of $-1$. Consider the allocations $X^c$ and $X^0$, where $X_1^c = \emptyset$ and $X_1^0 = \{o_1\}$.

Note that $X^c$ is clean with respect to $\{\beta^c_h\}_{h \in N}$ and $X^c \cup X^0$ is clean with respect to $\{\beta^0_h\}_{h \in N}$.
The singleton path $(o_2)$ is one of the shortest paths from $F_{\beta^c_1}(X^c)$ to $X^c_0$ in the exchange graph $\cal G(X^c, \beta^c)$. Augmenting along this path creates an allocation $Y$ where $Y^c_1 = \{o_2\}$ and $X^0_1 = \{o_1\}$.

Validating the correctness of \Cref{lem:path-augmentation}, $Y^c$ is indeed clean with respect to $\{\beta^c_h\}_{h \in N}$. However, $Y^c \cup X^0$ is not clean with respect to $\{\beta^0_h\}_{h \in N}$. 
Note that if we instead chose to augment along the singleton path $(o_1)$ instead of $(o_2)$, we would have not faced this issue.
\end{example}

In the above example, note that $X^c$ and $X^0$ do not form a valid decomposition of $X$. This is deliberate; in our algorithm design, we will not assume that $X^c$, $X^0$ and $X^{-1}$ form a valid decomposition of $X$. We will only assume that $X^c$ is clean with respect to $\{\beta^c_h\}_{h \in N}$ and $X^c \cup X^0$ is clean with respect to $\{\beta^0_h\}_{h \in N}$. Our goal is to use path augmentations to modify $X^c$ while retaining both these useful properties. 
To guarantee that $X^c \cup X^0$ remains clean (w.r.t. $\{\beta^0_h\}_{h \in N}$) even after path augmentation, we present a technique to carefully choose paths in the exchange graph. 
On a high level, this is done by giving weights to the edges in the exchange graph and choosing the least-weight path as opposed to a shortest path. 
The way we weigh edges is motivated by \Cref{ex:classic-paths-do-not-work} --- for all $i \in N$, we give edges from $X^c_i$ to $X^0_i$ a lower weight than other edges. It turns out that this simple change is sufficient to ensure the cleanness of $X^c \cup X^0$ is clean with respect to $\{\beta^0_h\}_{h \in N}$.

More formally, we define the weighted exchange graph $\cal G^w(X^c, X^0, \beta^c)$ as a weighted directed graph with the same nodes and edges as $\cal G(X^c, \beta^c)$. Each edge has a specific weight defined as follows: all edges from $o \in X^c_i$ to $o' \in X^0_i$ for any $i \in N$ are given a weight of $\frac12$; the remaining edges are given a weight of $1$. 

For any agent $i \in N$, we define two paths on this weighted exchange graph. A {\em Pareto-improving path} is a path from $F_{\beta^c_i}(X^c, i)$ for some $i \in N$ to some item in $X^c_0$. An {\em exchange path} is a path from $F_{\beta^c_i}(X^c, i)$ for some $i \in N$ to some item in $X^c_j$ for some $j \in N - i$.
We first show that path augmentation along the least-weight Pareto-improving path maintains the cleanness of $X^c \cup X^0$.

\begin{restatable}[Pareto Improving Paths]{theorem}{thmparetoimprovingpaths}\label{thm:Pareto-improving-paths}
When agents have $\{-1, 0, c\}\ONSUB$ valuations, let $X^c$ be a clean allocation with respect to $\{\beta^c_h\}_{h \in N}$ and $X^0$ be an allocation such that $X^c\cup X^0$ is clean with respect to the valuations $\{\beta^0_h\}_{h \in N}$. Let $X^c_h \cap X^0_h = \emptyset$ for all $h \in N$. For some agent $i \in N$, if a Pareto-improving path exists from $F_{\beta^c_i}(X^c, i)$ to $X^c_0$ in the weighted exchange graph $\cal G^w(X^c, X^0, \beta^c)$, then, path augmentation along the least-weight Pareto-improving path $P = (o_1, \dots, o_t)$ from $F_{\beta^c_i}(X^c, i)$ to $X^c_0$ in the weighted exchange graph $\cal G^w(X^c, X^0, \beta^c)$ results in the following allocations $Y^c$ and $Y^0$:
\begin{align*}
        Y^c_k = 
        \begin{cases}
            X^c_k \Uplambda P & (k \in N + 0 - i) \\
            X^c_i \Uplambda P + o_1 & (k = i)
        \end{cases} 
        ,
        &&
        Y^0_k = 
        \begin{cases}
            X^0_k - o_t & (k \in N) \\
            X^0_0 + o_t  & (k = 0)
        \end{cases}
        .
\end{align*}
$Y^c$ is clean with respect to with respect to $\{\beta^c_h\}_{h \in N}$ and $Y^c \cup Y^0$ is clean with respect to $\{\beta^0_h\}_{h \in N}$. Furthermore, for all $h \in N$, $Y^c_h \cap Y^0_h = \emptyset$.
\end{restatable}
\begin{proof}
Note that the path $P$ is still a shortest path in the (unweighted) graph $\cal G(X^c, \beta^c)$. This is because if there was a shorter path $P'$ of length $< t-1$, then this path $P'$ will have weight at most $t-2$ in the weighted exchange graph $\cal G^w(X^c, X^0, \beta^c)$ since each edge weight is upper bounded at $1$. On the other hand, the path $P$ has weight at least $t-2 + \frac12$, since there can be at most one edge with weight $\frac12$; all the other edges must have weight $1$. This is because all the $\frac12$ weight edges are incident on some item in $X^c_0$ and there are no edges coming out of any item in $X^c_0$ in both $\cal G(X^c, \beta^c)$ and $\cal G^w(X^c, X^0, \beta^c)$.

Therefore, path augmentation with the path $P$ will result in $Y^c$ being clean with respect to $\{\beta^c_h\}_{h \in N}$ (Lemma \ref{lem:path-augmentation}). In the rest of this proof, we show that $Y^c \cup Y^0$ is clean with respect to $\{\beta^0_h\}_{h \in N}$. Path augmentation with the path $P = (o_1, o_2, \dots, o_t)$ can be seen as a combination of single-item transfers done one after another. We first transfer $o_1$ from an agent (say $i_1$) to $i$, then we transfer $o_2$ from some other agent $i_2$ to $i_1$ and so on until we finally transfer $o_t$ from $X^c_0$ to some agent $i_t$. Note that $o_t$ is the only item from $\{o_1, \dots, o_t\}$ which is in $X^c_0$ since no item in $X^c_0$ has any outgoing edges.

We prove the lemma by induction. Let $P_l = (o_1, o_2, \dots, o_l)$ and let $Z^c, Z^0$ be the allocations that result from augmenting $X$ along the path $P_l$. Since $P$ is a shortest path from $o_1$ to $o_t$ in $\cal G(X^c, \beta^c)$, $P_l$ must be a shortest path from $o_1$ to $o_l$. Therefore, $Z^c$ is clean with respect to $\{\beta^c_h\}_{h \in N}$. We show for any $l < t$, if $Z^c \cup Z^0$ is clean with respect to $\{\beta^0_h\}_{h \in N}$, then moving $o_{l+1}$ from the agent who currently has it (say $j$) to the agent who was allocated $o_l$ (say $k$) still maintains the cleanness of $Z^c \cup Z^0$. Note that since $P_{l+1}$ is a shortest path from $o_1$ to $o_{l+1}$ in $\cal G(X^c, \beta^c)$ as well, $\Delta_{\beta^c_k}(Z^c_k, o_{l+1}) = 1$ (Lemma \ref{lem:path-augmentation}).

\noindent{\textbf{Case 1:}  We transfer an item $o_{l+1} \notin Z^0_k$ from $Z^c_j$ (for some $j \in N+0$) to $k$.}

Let us first study the sorted telescoping sum vector $\vec v_k(Z^c_k \cup Z^0_k)$. Define an ordering $\pi$ where the items in $Z^c_k$ appear first in an arbitrary order followed by the items in $Z^0_k$ in any arbitrary order. If there are more than $|Z^c_k|$ indices with value $c$ in $\vec v_k(Z^c_k \cup Z^0_k)$, there must be an item $o'\in Z^0_k$ such that $\Delta_{\beta^c_k}(Z^c_k, o') = 1$. This contradicts the fact that $P$ is the least-weight Pareto-improving path from $F_{\beta^c_i}(X^c_i, i)$ to $X^c_0$ since the path $(o_1, \dots, o_l, o')$ has strictly lesser weight than $P$. Therefore, the sorted telescoping sum vector $\vec v_k(Z^c_k \cup Z^0_k)$ must have exactly $|Z^c_k|$ indices with the value $c$. Since the bundle is clean with respect to $\{\beta^0_h\}_{h \in N}$, the other $|Z^0_k|$ indices must have value $0$ in $\vec v_k(Z^c_k \cup Z^0_k)$.

Let us now study the sorted telescoping sum vector $\vec v_k((Z^c_k \cup Z^0_k) + o_{l+1})$. Define an ordering $\pi'$ where the items in $Z^c_k$ appear first in an arbitrary order followed by the item $o_{l+1}$ followed by the items in $Z^0_k$ in any arbitrary order. In this ordering, we know that the sorted telescoping sum vector $\vec v_k((Z^c_k \cup Z^0_k) + o_{l+1})$ must have at least $|Z^c_k| + 1$ indices with a value of $c$ since $\beta^c_k(Z^c_k + o_{l+1}) = |Z^c_k| + 1$. 

Define a second ordering $\pi''$ where the items in $Z^c_k$ appear first in an arbitrary order followed by the items in $Z^0_k$ in any arbitrary order followed by the item $o_{l+1}$. In this ordering, since we showed that there are $|Z^0_k|$ indices with value $0$ in $\vec v_k(Z^c_k \cup Z^0_k)$, there are at least $|Z^0_k|$ indices with a value of $0$ in the sorted telescoping sum vector $\vec v_k((Z^c_k \cup Z^0_k) + o_{l+1})$. 

Combining these two observations, we get that there are $|Z^c_k| + 1$ indices with value $c$ and $|Z^0_k|$ indices with a value of $0$ in $\vec v_k((Z^c_k \cup Z^0_k) + o_{l+1})$. Therefore, $Z^c_k + o_{l+1}$ is clean with respect to $\beta^c_k$ and $(Z^c_k \cup Z^0_k) + o_{l+1}$ is clean with respect to $\beta^0_k$. Note that $Z^c_j - o_{l+1}$ is still clean with respect to $\beta^c_j$ and using submodularity, $(Z^c_j \cup Z^0_j) - o_{l+1}$ is clean with respect to $\beta^0_j$. 

\noindent{\textbf{Case 2:} We transfer an item $o_{l+1} \in X^0_k$ from $X^c_0$ to $k$.}
This case is trivial. $Z^c_k + o_{l+1}$ is clean with respect to $\beta^c_k$ by our choice of $o_{l+1}$ and since this is the final node in the path, $(Z^c_k+o_{l+1}) \cup (Z^0_k - o_{l+1}) = Z^c_k \cup Z^0_k$ is clean with respect to $\beta^0_k$ by assumption.
\end{proof}

In the above Theorem, we also modify $Y^0$ as part of the path augmentation operation to ensure $o_t$ is not present in two different bundles.
 We also show that when Pareto-improving paths do not exist, the least-weight exchange path maintains the cleanness of $X^c \cup X^0$. 

\begin{restatable}[Exchange Paths]{theorem}{thmexchangepaths}\label{thm:exchange-paths}
When agents have $\{-1, 0, c\}\ONSUB$ valuations, let $X^c$ be a clean allocation with respect to $\{\beta^c_h\}_{h \in N}$ and $X^0$ be an allocation such that $X^c\cup X^0$ is clean with respect to the valuations $\{\beta^0_h\}_{h \in N}$. Let $X^c_h \cap X^0_h = \emptyset$ for all $h \in N$. If no Pareto-improving path exists from $F_{\beta^c_i}(X, i)$ to $X^c_0$ in the weighted exchange graph $\cal G^w(X^c, X^0, \beta^c)$, then, path augmentation along the least-weight exchange path $P = (o_1, \dots, o_t)$ from $F_{\beta^c_i}(X, i)$ to $X^c_j$ (for any $j \in N - i$) in the weighted exchange graph $\cal G^w(X^c, X^0, \beta^c)$ results in the allocations $Y^c$ and $Y^0$:
\begin{align*}
        Y^c_k = 
        \begin{cases}
            X^c_k \Uplambda P & (k \in N + 0 - i) \\
            X^c_i \Uplambda P + o_1 & (k = i)
        \end{cases} 
        ,
        &&
        Y^0_k = X^0_k \qquad (k \in N)
        .
\end{align*}
$Y^c$ is clean with respect to with respect to $\{\beta^c_h\}_{h \in N}$ and $Y^c \cup Y^0$ is clean with respect to with respect to $\{\beta^0_h\}_{h \in N}$. Furthermore, for all $h \in N$, $Y^c_h \cap Y^0_h = \emptyset$.
\end{restatable}
\begin{proof}
This proof is very similar to Theorem \ref{thm:Pareto-improving-paths}. We can show that path $P$ is still a shortest path in the (unweighted) graph $\cal G(X^c, \beta^c)$ by a similar argument. For the same reason, path augmentation with the path $P$ will result in $Y^c$ being clean with respect to $\{\beta^c\}_{h \in N}$ (Lemma \ref{lem:path-augmentation}). In this proof, we show that $Y^c \cup Y^0$ is clean with respect to $\{\beta^0\}_{h \in N}$. 

We prove the lemma by induction as well. Let $P_l = (o_1, o_2, \dots, o_l)$ and let $Z^c, Z^0$ be the allocations that result from augmenting $X$ along the path $P_l$. Since $P$ is a shortest path from $o_1$ to $o_t$ in $\cal G(X^c, \beta^c)$, $P_l$ must be a shortest path from $o_1$ to $o_l$. Therefore, $Z^c$ is clean with respect to $\{\beta^c_h\}_{h \in N}$. We show for any $l < t$, if $Z^c \cup Z^0$ is clean with respect to $\{\beta^0_h\}_{h \in N}$, then moving $o_{l+1}$ from the agent who currently has it (say $j$) to the agent who was allocated $o_l$ (say $k$) still maintains the cleanness of $Z^c \cup Z^0$. 

Since $P_{l+1}$ is a shortest path from $o_1$ to $o_{l+1}$ in $\cal G(X^c, \beta^c)$ as well, $\Delta_{\beta^c_k}(Z^c_k, o_{l+1}) = 1$. Further, since this is an exchange path, $j \in N$.

Let us first study the sorted telescoping sum vector $\vec v_k(Z^c_k \cup Z^0_l)$. Define an ordering $\pi$ where the items in $Z^c_k$ appear first in an arbitrary order followed by the items in $Z^0_k$ in any arbitrary order. If there are more than $|Z^c_k|$ indices with value $c$ in $\vec v_k(Z^c_k \cup Z^0_k)$, there must be an item $o'\in Z^0_k$ such that $\Delta_{\beta^c_k}(Z^c_k, o') = 1$. This contradicts the fact that there is no Pareto-improving path in $\cal G^w(X^c, X^0, \beta^c)$ since the path $(o_1, \dots, o_l, o')$ is Pareto-improving. Therefore, the sorted telescoping sum vector $\vec v_k(Z^c_k \cup Z^0_k)$ must have exactly $|Z^c_k|$ indices with the value $c$. Since the bundle is clean with respect to $\{\beta^0\}_{h \in N}$, the other $|Z^0_k|$ indices must have value $0$ in $\vec v_k(Z^c_k \cup Z^0_k)$.

Let us now study the sorted telescoping sum vector $\vec v_k((Z^c_k \cup Z^0_k) + o_{l+1})$. Define an ordering $\pi'$ where the items in $Z^c_k$ appear first in an arbitrary order followed by the item $o_{l+1}$ followed by the items in $Z^0_k$ in any arbitrary order. In this ordering, we know that the sorted telescoping sum vector $\vec v_k((Z^c_k \cup Z^0_k) + o_{l+1})$ must have at least $|Z^c_k| + 1$ indices with a value of $c$ since $\beta^c_k(Z^c_k + o_{l+1}) = |Z^c_k| + 1$. 

Define a second ordering $\pi''$ where the items in $Z^c_k$ appear first in an arbitrary order followed by the items in $Z^0_k$ in any arbitrary order followed by the item $o_{l+1}$. In this ordering, since we showed that there are $|Z^0_k|$ indices with value $0$ in $\vec v_k(Z^c_k \cup Z^0_k)$, there are at least $|Z^0_k|$ indices with a value of $0$ in the sorted telescoping sum vector $\vec v_k((Z^c_k \cup Z^0_k) + o_{l+1})$. 

Combining these two observations, we get that there are $|Z^c_k| + 1$ indices with value $c$ and $|Z^0_k|$ indices with a value of $0$ in $\vec v_k((Z^c_k \cup Z^0_k) + o_{l+1})$. Therefore, $Z^c_k + o_{l+1}$ is clean with respect to $\beta^c_k$ and $(Z^c_k \cup Z^0_k) + o_{l+1}$ is clean with respect to $\beta^0_k$. Note that $Z^c_j - o_{l+1}$ is still clean with respect to $\beta^c_j$ and using submodularity, $(Z^c_j \cup Z^0_j) - o_{l+1}$ is clean with respect to $\beta^0_j$. 
\end{proof}

Note that since $\cal G^w(X^c, X^0, \beta^c)$ and $\cal G(X^c, \beta^c)$ have the same set of edges, Lemma \ref{lem:augmentation-sufficient} applies to the weighted exchange graph as well.

\begin{lemma}\label{lem:augmentation-sufficient-weighted}
When agents have $\{-1, 0, c\}\ONSUB$ valuations, let $X^c$ be a clean allocation with respect to $\{\beta^c_h\}_{h \in N}$ and $X^0$ be an allocation such that $X^c\cup X^0$ is clean with respect to the valuations $\{\beta^0_h\}_{h \in N}$. Let $X^c_h \cap X^0_h = \emptyset$ for all $h \in N$. For any $i \in N$ and $j \in N+0$, there is a path from $F_{\beta^c_i}(X^c, i)$ to $X^c_j$ in $\cal G^w(X^c, X^0, \beta^c)$ if and only if there is a path from $F_{\beta^c_i}(X^c, i)$ to $X^c_j$ in $\cal G(X^c, \beta^c)$.
\end{lemma}

Given an allocation $X = X^c \cup X^0 \cup X^{-1}$, these results already hint at a method to modify $X^c$ such that it becomes a leximin allocation with respect to the valuations $\{\beta^c_h\}_{h \in N}$: greedily use path augmentations in the weighted exchange graph until $X^c$ is leximin. 
This is the high-level approach we use in our algorithm.

\section{Leximin Allocations with $\{-1, 0, c\}$\ONSUB Valuations}\label{sec:leximin}
We are ready to present our algorithm to compute leximin allocations. Our algorithm has three phases. In the first phase, we use Yankee Swap (Theorem \ref{thm:yankee-swap}) to compute a \MAXUSW allocation $X$ with respect to the valuations $\{\beta^0_i\}_{i \in N}$. We also initialize $X^c$ and $X^{-1}$ to be empty allocations.

In the second phase, we update $X^c$ and $X^0$ using path augmentations (from Section \ref{sec:weighted-exchange-graph}) until $X^c$ is a leximin partial allocation. This is done by greedily augmenting along min weight Pareto-improving paths and min weight exchange paths until the allocation is leximin.

In the third phase, we update $X^{-1}$ by allocating the remaining items (in $X^c_0 \cap X^0_0$). We do so greedily by allocating each item to the agent with the highest utility, under the assumption that these items have a marginal utility of $-1$.
The exact steps are described in \Cref{algo:leximin}.

\setlength{\textfloatsep}{10pt}
\begin{algorithm}[!t]
    \caption{Leximin Allocations with $\{-1, 0, c\}$\ONSUB Valuations}
    \label{algo:leximin}
    \SetAlgoVlined
    \SetKwInOut{Input}{Input}
    \SetKwInOut{Output}{Output}
    \Input{A set of items $O$ and a set of agents $N$ with $\{-1, 0, c\}$\ONSUB valuations $\{v_h\}_{h \in N}$}
    \Output{A complete leximin allocation}
    \SetKwRepeat{Do}{repeat}{while}
    \LinesNumbered
    \DontPrintSemicolon
    \tcp{Phase 1: Make $X^c \cup X^0$ a \MAXUSW allocation w.r.t. $\{\beta^0_h\}_{h \in N}$}
    $X^0 \gets $ the output of Yankee Swap with respect to $\{\beta^0_h\}_{h \in N}$\;
    $X^{c} = (X^{c}_0, \dots, X^{c}_n) \gets (O, \emptyset, \dots, \emptyset)$\;
    \tcp*[r]{$X^c$ is clean w.r.t.\ $\beta^c$ and $X^c \cup X^0$ is clean w.r.t.\ $\beta^0$}
    $X^{-1} = (X^{-1}_0, \dots, X^{-1}_n) \gets (O, \emptyset, \dots, \emptyset)$\;

    \tcp{Phase 2: Make $X^c$ a clean leximin allocation w.r.t. $\{\beta^c_h\}_{h \in N}$}
    \Do{at least one path augmentation was done in the iteration}{
    \While{for some $i \in N$, there exists a Pareto-improving path from $F_{\beta^c_i}(X^c, i)$ in $\cal G^w(X^c, X^0, \beta^c)$}{
        $P = (o'_{1}, \dots, o'_{t}) \gets$ a min weight Pareto-improving path from $F_{\beta^c_i}(X^c, i)$ to $X^c_0$ in $\cal G^w(X^c, X^0, \beta^c)$ \; \tcc*[r]{Augment the allocation with the path $P$}
        $X^c_k \gets  X^c_k \Uplambda P$ for all $k \in N + 0 - i$\;
        $X^c_i \gets X^c_i \Uplambda P + o'_{1}$\;
        $X^0_k \gets X^0_k - o'_{t}$ for all $k \in N$\;
        $X^0_0 \gets X^0_0 + o'_{t}$\;
    }
    \If{for some $i \in N$, there exists an exchange path from $F_{\beta^c_i}(X^c, i)$ to some $X^c_j$ such that either (a) $|X^c_i| < |X^c_j| + 1$ or (b) $|X^c_i| = |X^c_j| + 1$ and $i < j$}{
        $P = (o'_{1}, \dots, o'_{t}) \gets$ a min weight exchange path from $F_{\beta^c_i}(X^c, i)$ to $X^c_j$ in $\cal G^w(X^c, X^0, \beta^c)$ \tcc*[r]{Augment the allocation with the path $P$}
        $X^c_k \gets  X^c_k \Uplambda P$ for all $k \in N + 0 - i$\;
        $X^c_i \gets X^c_i \Uplambda P + o'_{1}$\;
    }
    }
    \tcp{Phase 3: Greedily allocate the items in $X^c_0 \cap X^0_0$}
    \While{$|X^c_0 \cap X^0_0 \cap X^{-1}_0| > 0$}{ 
    \vspace{-0.36cm}
    \tcc*[r]{Unallocated items exist}
    \vspace{0.05cm}
        $S \gets \argmax\limits_{h \in N} v_h(X^c_h \cup X^0_h \cup X^{-1}_h)$ \tcc*{Set of all max-utility agents}
        \vspace*{-0.11cm}
        $i \gets \max\limits_{j \in S} j$ \tcc*{Break ties using index}
        \vspace*{-0.11cm}
        $o \gets $ an arbitrary item in $X^c_0 \cap X^0_0 \cap X^{-1}_0$\;
        $X^{-1}_i \gets X^{-1}_i + o$\; 
        $X^{-1}_0 \gets X^{-1}_0 - o$\;
    }
    \Return $X^c \cup X^0 \cup X^{-1}$\;
\end{algorithm}
We analyze each phase separately and establish key properties that the allocations $X^c, X^0$ and $X^{-1}$ have at the end of each phase. In order to show computational efficiency, we use the value oracle model where we have oracle access to each agent's valuation function. A computationally efficient algorithm runs in polynomial time (in $n$ and $m$) and only uses a polynomial number of queries to each value oracle.

\subsection{Phase 1}
The first phase is a simple setup where we allocate as many non-negative valued items as possible. $X^0$ is initialized as a clean \MAXUSW allocation with respect to the valuations $\{\beta^0_h\}_{h \in N}$. $X^c$ and $X^{-1}$ are initialized as the empty allocation. Note that $X^c$ is trivially clean with respect to the valuations $\{\beta^c_h\}_{h \in N}$.

\begin{obs}\label{obs:phase-1}
At the end of Phase 1, $X^c$ is clean with respect to $\{\beta^c_h\}_{h \in N}$ and $X^c \cup X^0$ is a \MAXUSW clean allocation with respect to $\{\beta^0_h\}_{h \in N}$.
\end{obs}

The computational efficiency of Phase 1 relies on the computational efficiency of Yankee Swap. Yankee Swap uses a polynomial number of queries to $\{\beta^0_h\}_{h \in N}$. We can easily construct an efficient oracle for each $\beta^0_h$ using Lemma \ref{lem:mrf-oracle} to ensure Phase 1 runs in polynomial time and a polynomial number of valuation queries.

\subsection{Phase 2}
In this phase, we use path augmentations to manipulate $X^c$ into a partial leximin allocation. There are three types of paths we check for and augment if it exists:
\begin{enumerate}[(a)]
    \item A Pareto-improving path from $F_{\beta^c_i}(X^c, i)$ for some $i \in N$ to $X^c_0$ in $\cal G^w(X^c, X^0, \beta^c)$.
    \item An exchange path from $F_{\beta^c_i}(X^c, i)$ for some $i \in N$ to some $X^c_j$ such that $|X^c_i| < |X^c_j| + 1$.
    \item An exchange path from $F_{\beta^c_i}(X^c, i)$ for some $i \in N$ to some $X^c_j$ such that $|X^c_i| = |X^c_j| + 1$ and $i < j$.
\end{enumerate}

Note that these path augmentations work as intended since we maintain the invariant that $X^c \cup X^0$ is clean with respect to $\{\beta^0_h\}_{h \in N}$. Since exchange paths ((b) and (c)) are only guaranteed to work when there are no Pareto-improving paths, we ensure that we augment Pareto-improving paths (a) first before we check and augment exchange paths. This is clear in the steps of Algorithm \ref{algo:leximin}.

Formally, let $Z^c$ be a clean leximin allocation with respect to the valuations $\{\beta^c_h\}_{h \in N}$. If there are multiple, let $Z^c$ be an allocation that is not lexicographically dominated (w.r.t. $\{\beta^c_h\}_{h \in N}$) by any other leximin allocation. Phase $2$ ensures that for each agent $i \in N$, $|X^c_i| = |Z^c_i|$. We have the following Lemma.

\begin{restatable}{lemma}{lemphasetwosufficient}\label{lem:phase-2-sufficient}
Let $X^c$ be a clean allocation with respect to $\{\beta^c_h\}_{h \in N}$ and $X^0$ be an allocation such that $X^c \cup X^0$ is clean with respect to $\{\beta^0_h\}_{h \in N}$. Let $Z^c$ be a clean leximin allocation with respect to the valuations $\{\beta^c_h\}_{h \in N}$. If there are multiple, let $Z^c$ be an allocation that  is not lexicographically dominated (w.r.t. $\{\beta^c_h\}_{h \in N}$) by any other leximin allocation. Then there exists an agent $\ell \in N$ such that $|X^c_{\ell}| \ne |Z^c_{\ell}|$ if and only if at least one of the following conditions hold:
\begin{enumerate}[(a),nosep]
    \item There exists a Pareto-improving path from $F_{\beta^c_i}(X^c, i)$ for some $i \in N$ to $X^c_0$ in $\cal G^w(X^c, X^0, \beta^c)$.
    \item There exists an exchange path from $F_{\beta^c_i}(X^c, i)$ for some $i \in N$ to some $X^c_j$ such that $|X^c_i| < |X^c_j| + 1$.
    \item There exists an exchange path from $F_{\beta^c_i}(X^c, i)$ for some $i \in N$ to some $X^c_j$ such that $|X^c_i| = |X^c_j| + 1$ and $i < j$.
\end{enumerate}
\end{restatable}
\begin{proof}

\noindent\textbf{Part 1:} Assume $|X^c_{\ell}| = |Z^c_{\ell}|$ for all $\ell \in N$.

If there exists a Pareto-improving path from some $F_{\beta^c_i}(X^c, i)$ to $X^c_0$ in $\cal G^w(X^c, X^0, \beta^c)$, augmenting along the least-weight path will result in an allocation $Y^c$ which Pareto dominates $Z^c$ with respect to the valuations $\{\beta^c_{h}\}_{h \in N}$. This contradicts the assumption that $Z^c$ is leximin with respect to the valuations $\{\beta^c_{h}\}_{h \in N}$.

Similarly, if there exists an exchange path from some $F_{\beta^c_i}(X^c, i)$ to some $X^c_j$ such that $|X^c_i| < |X^c_j| + 1$, augmenting along the least-weight path will result in an allocation $Y^c$ such that the sorted utility vector of $Y^c$ with respect to the valuations $\{\beta^c_{h}\}_{h \in N}$ lexicographically dominates the sorted utility vector of $Z^c$. This again contradicts the assumption that $Z^c$ is leximin with respect to the valuations $\{\beta^c_{h}\}_{h \in N}$.

Finally, if there exists an exchange path from some $F_{\beta^c_i}(X^c, i)$ to some $X^c_j$ such that $|X^c_i| = |X^c_j| + 1$ and $i < j$, augmenting along the least-weight path will result in an allocation $Y^c$ with the same sorted utility vector as $Z^c$ with respect to the valuations $\{\beta^c_{h}\}_{h \in N}$ but $\vec u^{Y^c}$ lexicographically dominates $\vec u^{Z^c}$. This again contradicts the assumption that $Z^c$ is not lexicographically dominated by any other leximin allocation. 

\noindent\textbf{Part 2:} Assume $|X^c_{\ell}| \ne |Z^c_{\ell}|$ for some $\ell \in N$.

Assume for contradiction that none of the paths (a), (b) and (c) exist in the weighted exchange graph $\cal G^w(X^c, X^0, \beta^c)$.

Let $i \in N$ be the agent with least $|X^c_i|$ such that $|X^c_i| < |Z^c_i|$; break ties by choosing the agent with the least index $i$. If no agent $i$ exists, then we again contradict the fact that $Z^c$ is a leximin allocation with respect to the valuations $\{\beta^c_{h}\}_{h \in N}$. If no Pareto-improving paths exist, using Theorem \ref{lem:augmentation-sufficient}, we get that there must be an agent $j \in N$ such that $|Z^c_j| < |X^c_j|$. If there are multiple such agents, choose the one with the least $|Z^c_j|$; break further ties by choosing the agent with the least index $j$.
We have the following $4$ cases, each leading to a contradiction. 

\noindent\textbf{Case 1:} $|Z^c_j| < |X^c_i|$. From Lemma \ref{lem:augmentation-sufficient}, there must be a path from $j$ to some $p$ such that $|Z^c_p| > |X^c_p|$ in the exchange graph $\cal G(Z^c, \beta^c)$. Augmenting along a shortest such path gives us an allocation $\hat{Z}^c$. Since $|X^c_p| \ge |X^c_i|$ by our choice of $i$, we have $|\hat{Z}^c_p| \ge |X^c_p| > |Z^c_j|$ and $|\hat{Z}^c_j| > |Z^c_j|$. Since $p$ and $j$ are the only two agents whose utilities (w.r.t. $\{\beta^c_h\}_{h \in N}$) differ in allocations $\hat{Z}^c$ and $Z^c$, we can conclude that the sorted utility vector of $\hat Z^c$ (w.r.t. $\{\beta^c_h\}_{h \in N}$) lexicographically dominates the sorted utility vector of $Z^c$ --- contradicting the fact that $Z^c$ is leximin.

\noindent\textbf{Case 2:} $|Z^c_j| = |X^c_i|$ and $j < i$.
This case is similar to the previous case. From Lemma \ref{lem:augmentation-sufficient}, there must be a path from $j$ to some $p$ such that $|Z^c_p| > |X^c_p|$. Augmenting along a shortest such path gives us an allocation $\hat{Z}^c$. Note that by our choice of $i$, $|X^c_p| \ge |X^c_i|$ and if equality holds, $p > i$. We have 
\begin{align*}
    |\hat{Z}^c_p| \ge |X^c_p| \ge |X^c_i| = |Z^c_j|.
\end{align*}
If any of these weak inequalities are strict, we can use the analysis from Case 1 to show a contradiction. If equality holds throughout, $\hat{Z}^c$ and $Z^c$ have the same sorted utility vectors. However, since $j < i < p$, $\hat{Z}^c$ lexicographically dominates $Z^c$ (w.r.t. $\{\beta^c_h\}_{h \in N}$) --- again, creating a contradiction.

\noindent\textbf{Case 3:} $|Z^c_j| > |X^c_i|$. From Lemma \ref{lem:augmentation-sufficient}, there must be a path from $i$ to some $q$ such that $|Z^c_q| < |X^c_q|$ in the exchange graph $\cal G(X^c, \beta^c)$. From Lemma \ref{lem:augmentation-sufficient-weighted}, there must be such a path in the weighted exchange graph $\cal G(X^c, X^0, \beta^c)$. Since $|Z^c_q| \ge |Z^c_j|$, we have
\begin{align*}
    |X^c_q| > |Z^c_q| \ge |Z^c_j| > |X^c_i|,
\end{align*}
which implies $|X^c_q| > |X^c_i| + 1$ which contradicts our assumption that there is no path of type (b) in the weighted exchange graph.

\noindent\textbf{Case 4:} $|Z^c_j| = |X^c_i|$ and $j > i$. This is similar to the previous case. From Lemma \ref{lem:augmentation-sufficient}, there must be a path from $i$ to some $q$ such that $|Z^c_q| < |X^c_q|$ in the exchange graph $\cal G(X^c, \beta^c)$. From Lemma \ref{lem:augmentation-sufficient-weighted}, there must be such a path in the weighted exchange graph $\cal G(X^c, X^0, \beta^c)$. Note that $|Z^c_q| \ge |Z^c_j|$ and equality holds only if $q > j$. Since $|Z^c_q| + 1\le |X^c_q|$ we have
\begin{align*}
    |X^c_q| \ge |Z^c_q| + 1 \ge |Z^c_j| + 1 = |X^c_i| + 1.
\end{align*}
If any of the weak inequalities are strict, we have $|X^c_q| > |X^c_i| + 1$ which contradicts our assumption that there is no path of type (b) in the weighted exchange graph.

If all the weak inequalities are equalities, we have $|X^c_q| + 1 = |X^c_i|$ and $q > j > i$ which contradicts our assumption that there are no paths of type (c) in the weighted exchange graph.
\end{proof}

Lemma \ref{lem:phase-2-sufficient} immediately implies that {\em if} Phase 2 terminates, it terminates when $|X^c_h| = |Z^c_h|$ for all $h \in N$. Our next Lemma shows that Phase 2 indeed terminates in polynomial time. This result follows from bounding the number of path augmentations our algorithm computes.

\begin{restatable}{lemma}{lemphasetwoterminates}\label{lem:phase-2-terminates}
Phase 2 terminates in polynomial time and polynomial valuation queries.
\end{restatable}
\begin{proof}
Note that checking for the existence of a path and augmenting along a path can trivially be done in polynomial time and polynomial valuation queries. So it suffices to simply upper bound the number of paths we augment. 

Note that there can be at most $m$ augmentations of Pareto-improving paths. This is because after each augmentation, the value $\sum_{h \in N} |X^c_h|$ strictly increases by $1$ and $\sum_{h \in N} |X^c_h|$ is upper bounded at $m$. 

To bound the number of exchange path augmentations, we use the potential function\footnote{This is the same potential function used by \citet{Babaioff2021Dichotomous}.} 
\begin{align*}
  \Phi(X^c) = \sum_{h \in N} \bigg (|X^c_h| + \frac{h}{n^2} \bigg )^2 .
\end{align*}
Note that this value is upper bounded by $\mathcal{O}(m^2)$ and lower bounded at $0$. It is easy to verify that after each of the specific exchange path augmentations in \Cref{algo:leximin}, $\Phi(X^c)$ strictly decreases. Any strict decrease with this potential function is lower bounded by $\frac{1}{n^4}$ since $\Phi(X^c)$ is always a multiple of $\frac{1}{n^4}$. Therefore, there can be at most $\mathcal{O}(n^4m^2)$ exchange path augmentations that happen consecutively --- with no Pareto-improving path augmentation happening in between. 

After a Pareto-improving path augmentation, note that $\Phi(X^c)$ increases. Therefore, 
we can only guarantee that there will be $\mathcal{O}(n^4m^2)$ exchange paths in between any two consecutive Pareto-improving path augmentations. Since there can be at most $m$ Pareto-improving path transfers, assuming a worst case, we can upper bound the number of exchange path augmentations at $\mathcal{O}(n^4m^3)$.
\end{proof}

Note that after each path transfer the number of allocated items in $X^c \cup X^0$ weakly increases. Therefore, $X^c \cup X^0$ remains a clean \MAXUSW allocation at the end of Phase 2. Formally, we can make the following observation.
\begin{obs}\label{obs:phase-2}
At the end of Phase 2, $X^c$ is a clean lexicographically dominating leximin allocation with respect to the valuations $\{\beta^c_h\}_{h \in N}$ and $X^c \cup X^0$ is a \MAXUSW clean allocation with respect to $\{\beta^0_h\}_{h \in N}$.
\end{obs}

\subsection{Phase 3}
At this point, $X^c\cup X^0$ is a \MAXUSW clean allocation  with respect to $\{\beta^0_h\}_{h \in N}$. Thus, all of the remaining items have a marginal value of $-1$ to every agent. It remains to assign these items in as equitable a manner as possible. 
In Phase 3, we sequentially allocate the remaining items, giving a ``bad'' item to the agent with the highest utility. 

To carefully compare allocations, we adapt the comparison method of {\em domination} introduced in \citet{cousins2023bivalued}. 
To compare two allocations $Y = Y^c \cup Y^0 \cup Y^{-1}$ and $\hat Y = \hat{Y}^c \cup \hat{Y}^0 \cup \hat{Y}^{-1}$, we first check the sorted utility vectors $\vec s^{Y^c}$ and $\vec s^{\hat{Y}^c}$; 
if $\vec s^{Y^c}$ lexicographically dominates $\vec s^{\hat{Y}^c}$, then we say $Y$ dominates $\hat Y$. 
If $\vec s^{Y^c}$ and $\vec s^{\hat{Y}^c}$ are equal, we check the utility vectors $\vec u^{Y^c}$ and $\vec u^{\hat{Y}^c}$; 
if $\vec u^{Y^c}$ lexicographically dominates $\vec u^{\hat{Y}^c}$, then we say $Y$ dominates $\hat Y$. 
If $\vec u^{Y^c}$ and $\vec u^{\hat{Y}^c}$ are equal as well, we check $\vec u^{Y}$ and $\vec u^{\hat{Y}}$; if $\vec u^{Y}$ lexicographically dominates $\vec u^{\hat{Y}}$, then we say $Y$ dominates $\hat Y$.

\begin{definition}[Domination]\label{def:domination}
An allocation $Y = Y^c \cup Y^0 \cup Y^{-1}$ dominates the allocation $\hat Y = \hat{Y}^c \cup \hat{Y}^0 \cup \hat{Y}^{-1}$ if any of the following conditions hold:
  \begin{inparaenum}[(a)]
    \item $\vec s^{Y^c} \succ_{\lex} \vec s^{\hat{Y}^c}$,
    \item $\vec s^{Y^c} = \vec s^{\hat{Y}^c}$ and $\vec u^{Y^c} \succ_{\lex} \vec u^{\hat{Y}^c}$, or 
    \item $\vec u^{Y^c} = \vec u^{\hat{Y}^c}$ and $\vec u^Y \succ_{\lex} \vec u^{\hat Y}$.
\end{inparaenum}
\end{definition}

Let $Y = Y^c \cup Y^0 \cup Y^{-1}$ be a {\em complete} leximin allocation for the original instance with valuations $\{v_h\}_{h \in N}$. If there are multiple leximin allocations, pick one which is not dominated by any other leximin allocation $\hat Y$. Due to this specific choice, we sometimes refer to $Y$ as a dominating leximin allocation.

We first show that, just like $X^c$, $Y^c$ is a lexicographically dominating leximin allocation with respect to the valuations $\{\beta^c_h\}_{h \in N}$.

\begin{restatable}{lemma}{lemxcequalsyc}\label{lem:xc-equals-yc}
After the end of Phase $2$, $|X^c_h| =|Y^c_h|$ for all $h \in N$.
\end{restatable}
\begin{proof}
Assume for contradiction this is not true.
We know from \Cref{obs:phase-2} that $X^c$ is a clean leximin allocation with respect to the valuations $\{\beta^c_h\}_{h \in N}$. Therefore, using Lemma \ref{lem:phase-2-sufficient}, at least one of three paths must exist in the weighted exchange graph $\cal G(Y^c, Y^0, \beta^c)$.

\noindent\textbf{Case 1: } There exists a Pareto-improving path from some $F_{\beta^c_i}(Y^c, i)$ in $\cal G^w(Y^c, Y^0, \beta^c)$.

Let $P = (o_1, \dots, o_t)$ be the least-weight path in the weighted exchange graph from some $F_{\beta^c_i}(Y^c, i)$ to $Y^c_0$ in the exchange graph $\cal G^w(Y^c, Y^0, \beta^c)$. Path augmentation along this path $P$ gives us two allocations $\hat{Y}^c$ and $\hat{Y}^0$ such that $\hat{Y}^c$ is clean with respect to $\{\beta^c_h\}_{h \in N}$ and $\hat{Y}^c \cup \hat{Y}^0$ is clean with respect to $\{\beta^0_h\}_{h \in N}$. Define an allocation $\hat{Y}^{-1}$ as follows:
\begin{align*}
    \hat{Y}^{-1}_k = 
    \begin{cases}
        Y^{-1}_k - o_t & (k \in N) \\
        Y^{-1}_0 + o_t & (k = 0)
    \end{cases}
    .
\end{align*}

Define $\hat Y = \hat{Y}^c \cup \hat{Y}^0 \cup \hat{Y}^{-1}$. Note that with this construction, for all $h \in N$, $\hat{Y}^c_h$, $\hat{Y}^0_h$ and $\hat{Y}^{-1}_h$ are pairwise disjoint.
Moreover, $\hat{Y}^c$ is clean with respect to $\{\beta^c_h\}_{h \in N}$ and $\hat{Y}^c \cup \hat{Y}^0$ is clean with respect to $\{\beta^0_h\}_{h \in N}$. Therefore, we have for each agent $h \in N$, $v_h(\hat{Y}_h) \ge c|\hat{Y}^c_h| - |\hat{Y}^{-1}_h|$.

From the definition of $\hat{Y^c}$, we have $|\hat{Y}^c_h| \ge |Y^c_h|$ for every agent $h \in N$, with the inequality being strict for agent $i$ and $|\hat{Y}^{-1}_h| \le |Y^{-1}_h|$ for every agent $h \in N$. Combining these two inequalities, we have for every agent $h \in N$
\begin{align*}
    v_h(\hat{Y}_h) \ge c|\hat{Y}^c_h| - |\hat{Y}^{-1}_h| \ge c|Y^c_h| - |Y^{-1}_h| = v_h(Y_h),
\end{align*}
with at least one inequality being strict for agent $i$. This implies that the allocation $\hat{Y}$ Pareto dominates $Y$ (w.r.t. the valuations $\{v_h\}_{h \in N}$) which contradicts the assumption that $Y$ is leximin.

\noindent\textbf{Case 2: } There exists no Pareto-improving path from some $F_{\beta^c_i}(Y^c, i)$ in $\cal G^w(Y^c, Y^0, \beta^c)$.
 
From Lemma \ref{lem:phase-2-sufficient}, we get that there must be an exchange path from some $F_{\beta^c_i}(Y^c, i)$ to $Y^c_j$ for some $j \in N$ in $\cal G^w(Y^c, Y^0, \beta^c)$ such that $|Y^c_i| \le |Y^c_j| + 1$ with equality holding iff $i < j$. Path augmentation along the least-weight such path in the exchange graph $\cal G^w(Y^c, Y^0, \beta^c)$ results in allocations $\hat{Y}^c$ and $\hat{Y}^0$ such that $\hat{Y}^c$ is clean with respect to $\{\beta^c_h\}$ and $\hat{Y}^c \cup \hat{Y}^0$ is clean with respect to $\{\beta^0_h\}$. Note that the path does not contain any item from $Y^c_0$. Therefore, if we construct an allocation $\hat Y^{-1}$ such that $\bigcup_{h \in N} \hat{Y}^{-1}_h = \bigcup_{h \in N} Y^{-1}_h$, then $\hat{Y}^c \cup \hat{Y}^0 \cup \hat{Y}^{-1}$ is a complete allocation such that for each agent $h \in N$, $\hat{Y}^c_h$, $\hat{Y^0_h}$ and $\hat{Y}^{-1}_h$ are pairwise disjoint.
We have the following two subcases:

\noindent\textbf{Sub-case (a): } $|Y^c_i| < |Y^c_j| + 1$.

If $|Y^{-1}_j| < c$, consider the allocation $\hat{Y}$ defined as $\hat{Y}^c \cup \hat{Y}^0 \cup {Y}^{-1}$. We have the following inequalities for each agent $h \in N$:
\begin{enumerate}[(i)]
    \item $v_h(\hat{Y}_h) \ge c|\hat{Y}^c_h| - |Y^{-1}_h| = c|{Y}^c_h| - |Y^{-1}_h| = v_h(Y_h)$ for all $h \in N - i - j$,
    \item $v_i(\hat{Y}_i) \ge c|\hat{Y}^c_i| - |Y^{-1}_i| = c(|Y^c_i|+1) - |Y^{-1}_i| > v_i(Y_i)$, and
    \item $v_j(\hat{Y}_j) \ge c|\hat{Y}^c_j| - |Y^{-1}_j| = c(|Y^c_j|-1) - |Y^{-1}_j| \ge v_j(Y_j) - c$.
\end{enumerate}
In simple words, when comparing $\hat{Y}$ and $Y$, all the agents in $N-i-j$ weakly gain utility, $i$ strictly gains utility and $j$ loses utility by at most $c$. However, note that

\begin{align*}
v_j(\hat{Y}_j) \ge c(|Y^c_j|-1) - |Y^{-1}_j| \ge c(|Y^c_i| + 1) - |Y^{-1}_j| > c(|Y^c_i| + 1) - c \ge v_i(Y_i).    
\end{align*}

Therefore, the sorted utility vector of $\hat{Y}$ with respect to $\{v_h\}_{h \in N}$ lexicographically dominates the sorted utility vector of $Y$. 

If $|Y^{-1}_j| \ge c$, let $S$ be a $c$-sized subset of $Y^{-1}_j$. Consider the allocation $\hat{Y}^c \cup \hat{Y}^0 \cup \hat{Y}^{-1}$ where $\hat{Y}^{-1}$ is defined as follows:
\begin{align*}
    \hat{Y}^{-1}_k = 
    \begin{cases}
        Y^{-1}_k & (k \in N - i - j) \\
        Y^{-1}_j \setminus S & (k = j) \\
        Y^{-1}_i \cup S & (k = i)
    \end{cases}
    .
\end{align*}
Define $\hat Y = \hat{Y}^c \cup \hat{Y}^0 \cup \hat{Y}^{-1}$.
We can set up a similar set of inequalities for each agent $h \in N$.
\begin{enumerate}[(i)]
    \item $v_h(\hat{Y}^c_h) \ge c|\hat{Y}^c_h| - |\hat{Y}^{-1}_h| = c|{Y}^c_h| - |Y^{-1}_h| = v_h(Y_h)$ for all $h \in N - i - j$,
    \item $v_i(\hat{Y}_i) \ge c|\hat{Y}^c_i| - |\hat{Y}^{-1}_i| = c(|Y^c_i|+1) - |Y^{-1}_i| - c \ge v_i(Y_i)$, and
    \item $v_j(\hat{Y}_j) \ge c|\hat{Y}^c_j| - |\hat{Y}^{-1}_j| = c(|Y^c_j|-1) - |Y^{-1}_j| + c \ge v_j(Y_j)$.
\end{enumerate}
If any of these inequalities are strict, $\hat{Y}$ Pareto dominates $Y$ contradicting the fact that $Y$ is leximin. 

If all the inequalities are equalities, then $\hat{Y}$ and $Y$ have the same sorted utility vector, implying that they are both leximin. However, $\vec s^{\hat{Y}^c} \succ_{\lex} \vec s^{Y^c}$ contradicting our choice of $Y$ as a dominating leximin allocation.

\noindent\textbf{Sub-case (b): $|Y^c_i| = |Y^c_j| + 1$ and $i < j$}. Consider the allocation $\hat{Y}^c \cup \hat{Y}^0 \cup \hat{Y}^{-1}$ where $\hat{Y}^{-1}$ is defined as follows:
\begin{align*}
    \hat{Y}^{-1}_k = 
    \begin{cases}
        Y^{-1}_k & (k \in N - i - j) \\
        Y^{-1}_j  & (k = i) \\
        Y^{-1}_i & (k = j)
    \end{cases}
    .
\end{align*}
Essentially, we swap $Y^{-1}_j$ and $Y^{-1}_i$. Define $\hat Y = \hat{Y}^c \cup \hat{Y}^0 \cup \hat{Y}^{-1}$. We again compare the utilities of each agent. 
\begin{enumerate}[(i)]
    \item $v_h(\hat{Y}_h) \ge c|\hat{Y}^c_h| - |\hat{Y}^{-1}_h| = c|{Y}^c_h| - |Y^{-1}_h| = v_h(Y_h)$ for all $h \in N - i - j$,
    \item $v_i(\hat{Y}_i) \ge c|\hat{Y}^c_i| - |\hat{Y}^{-1}_i| = c(|Y^c_i|+1) - |Y^{-1}_j|  \ge v_j(Y_j)$, and
    \item $v_j(\hat{Y}_j) \ge c|\hat{Y}^c_j| - |\hat{Y}^{-1}_j| = c(|Y^c_j|-1) - |Y^{-1}_i| \ge v_i(Y_i)$.
\end{enumerate}

Again, if any of the above inequalities are strict, we have $\vec s^{\hat{Y}} \succ_{\lex} \vec s^{Y}$, contradicting the fact that $Y$ is leximin. 

Therefore, equality must hold throughout and $\hat{Y}$ must be a leximin allocation. This also means that $\hat{Y}^c$, $\hat{Y}^0$ and $\hat{Y}^{-1}$ must be a decomposition of $\hat{Y}$. This implies $\vec s^{\hat{Y}^c} = \vec s^{Y^c}$ and since $i < j$, $\vec u^{\hat{Y}^c} \succ_{\lex} \vec u^{Y^c}$. This again contradicts our choice of $Y$ as a dominating leximin allocation.
\end{proof}

Next, we show that $X^c$, $X^0$ and $X^{-1}$ form a valid decomposition of $X^c \cup X^0 \cup X^{-1}$.

\begin{restatable}{lemma}{lemalgodecomposition}\label{lem:algo-decomposition}
At every iteration in Phase $3$, for any agent $i \in N$, $v_i(X_i) = c|X^c_i| - |X^{-1}_i|$.
\end{restatable}
\begin{proof}
Note that $X^c$ is a clean leximin allocation with respect to $\{\beta^c\}_{h \in N}$.  Therefore, $v_i(X^c_i) = c|X^c_i|$.
If $v_i(X_i) > c|X^c_i| - |X^{-1}_i|$, it means that there are strictly more than $|X^c_i \cup X^0_i|$ indices in the telescoping sum vector $\vec v_i(X_i)$ with a non-negative value. The set of all items $S$ associated with these indices can be used to create a clean allocation with respect to $\{\beta^0\}_{h \in N}$ with a higher USW than $X^c \cup X^0$.

Specifically construct an allocation $Y$ starting at $X^c \cup X_0$ and replacing $X^c_i \cup X^0_i$ with $S$. This is clean (w.r.t. $\{\beta^0\}_{h \in N}$) by our choice of $S$ (using a similar argument to \Cref{lem:decomposition}). This allocation also has a higher USW than $X^c \cup X^0$ (w.r.t. $\{\beta^0_h\}_{h \in N}$) since $|S| > |X^c_i \cup X^0_i|$ --- a contradiction to \Cref{obs:phase-2}. 
\end{proof}

Combining these Lemmas, we can show the correctness of our algorithm.
\begin{restatable}{theorem}{thmleximin}\label{thm:leximin}
When agents have $\{-1, 0, c\}\ONSUB$ valuations, \Cref{algo:leximin} computes a leximin allocation efficiently.   
\end{restatable}
\begin{proof}
The computational efficiency of the Algorithm follows from \Cref{lem:phase-2-terminates}. Phases 1 and 3 are trivially efficient. So we only show correctness.

We use $X$ to denote the allocation $X^c \cup X^0 \cup X^{-1}$.
In Lemma \ref{lem:xc-equals-yc}, we show that for all agents $h \in N$, $|X^c_h| = |Y^c_h|$. Therefore, if we show that $|X^{-1}_h| = |Y^{-1}_h|$ for each agent $h \in N$, we prove that $X$ is leximin as well. 

Let $i \in N$ be the agent with least $v_i(X_i)$ such that $|X^{-1}_i| > |Y^{-1}_i|$; break ties by choosing the agent with the highest index $i$. If there is no such agent $i$, then $|X^{-1}_h| \le |Y^{-1}_h|$ for each agent $h \in N$, which implies $v_h(X_h) \ge v_h(Y_h)$ which means $X$ is leximin as well. 

From \Cref{obs:phase-2}, we have $\sum_{h \in N} |X^c_h \cup X^0_h| \ge |Y^c_h \cup Y^0_h|$. This means that $\sum_{h \in N} |X^{-1}_h| \le |Y^{-1}_h|$ since both $X$ and $Y$ are complete allocations. Therefore, there must exist some agent $j \in N$ such that $|X^{-1}_j| < |Y^{-1}_j|$.

Define $\hat{Y}^{-1}$ as an allocation starting at $Y^{-1}$ and moving an arbitrary item $o$ from $\hat{Y}^{-1}_j$ to $\hat{Y}^{-1}_i$. Consider the allocation $\hat{Y} = Y^c \cup Y^0 \cup \hat{Y}^{-1}$. 
If $\Delta_{v_i}(Y_i, o) \ge 0$, we are done. So assume $\Delta_{v_i}(Y_i, o) = -1$.

Let $W = X^c \cup X^0 \cup W^{-1}$ be the non-redundant allocation maintained by Algorithm \ref{algo:leximin} at the start of the iteration in Phase 3 where $i$ received its final item. By our choice of iteration, we must have 
\begin{align}
    v_i(W_i) \ge v_j(W_j).\text{ If equality holds, then } i > j  . \label{eq:wi-greater-wj} 
\end{align}

We can use this to compare $Y$ and $\hat{Y}$:
\begin{align}
    v_i(Y_i) = c|X^c_i| - |Y^{-1}_i| &\ge c|X^c_i| - |W^{-1}_i| = v_i(W_i) \ge v_j(W_j)  \notag\\
    &=  c|X^c_j| - |W^{-1}_j| \ge c|X^c_j| - |\hat{Y}^{-1}_j| = v_j(\hat{Y}_j).  \label{eq:yi-greater-zj}
\end{align}
The first equality follows from Lemma \ref{lem:algo-decomposition}.
If any of these weak inequalities are strict, we are done --- $v_i(Y_i) > v_j(\hat{Y}_j)$ implies $v_i(\hat{Y}_i) > v_j(Y_j)$ since the item $o$ has a marginal value of $-1$ in both bundles. We also trivially have $v_j(\hat{Y}_j) > v_j(Y_j)$. Since $j$ and $i$ were the only two agents who say a change in their bundles, this means $\vec s^{\hat{Y}} \succ_{\lex} \vec s^Y$ --- a contradiction.

If all the weak inequalities are equalities, then $v_i(Y_i) = v_j(\hat{Y}_j)$ which in turn implies that $v_i(\hat{Y}_i) = v_j(Y_j)$. This also implies that $\vec s^{\hat{Y}}= \vec s^Y$ which means both $Y$ and $\hat{Y}$ are leximin. However, since $i > j$ (from \eqref{eq:wi-greater-wj}), $\hat{Y}$ lexicographically dominates $Y$ (w.r.t. $\{v_h\}_{h \in N}$) which contradicts our choice of $Y$. 
\end{proof}

A useful corollary of this analysis is that Algorithm \ref{algo:leximin} outputs a \MAXUSW allocation. 

\begin{restatable}{corollary}{corrmaxusw}\label{corr:max-usw}
When agents have $\{-1, 0, c\}\ONSUB$ valuations, \Cref{algo:leximin} outputs a \MAXUSW allocation.
\end{restatable}
\begin{proof}
Let $X^c, X^0$ and $X^{-1}$ be the allocations output by Algorithm \ref{algo:leximin}. Define $X$ as $X^c \cup X^0 \cup X^{-1}$. Note that the USW of this allocation is given by $\sum_{i \in N} v_i(X_i) = \sum_{i \in N} c|X^c_i| - |X^{-1}_i|$ from Lemma \ref{lem:algo-decomposition}. 

Let $Y = Y^c \cup Y^0 \cup Y^{-1}$ be any other allocation. Since $X^c$ is a clean leximin allocation with respect to $\{\beta^c_h\}_{h \in N}$, it is also \MAXUSW with respect to  $\{\beta^c_h\}_{h \in N}$ \citep{Babaioff2021Dichotomous, viswanathan2022yankee, benabbou2021MRF}. Therefore, $\sum_{i \in N} |X^c_i| \ge \sum_{i \in N} |Y^c_i|$. 

Since $X^c \cup X^0$ is a \MAXUSW allocation with respect to $\{\beta^0_h\}_{h \in N}$, we must have $\sum_{i \in N} |X^{-1}_i| \le \sum_{i \in N} |Y^{-1}_i|$. Combining these two observations we have:
\begin{align*}
    \sum_{i \in N} v_i(X_i) &= \sum_{i \in N} c|X^c_i| - |X^{-1}_i| \\
    &= c \sum_{i \in N} |X^c_i| - \sum_{i \in N}  |X^{-1}_i| \\
    &\ge c \sum_{i \in N} |Y^c_i| - \sum_{i \in N}  |Y^{-1}_i| \\
    &\ge \sum_{i \in N} c|Y^c_i| - |Y^{-1}_i| = \sum_{i \in N} v_i(Y_i).
\end{align*}

Since we picked $Y$ arbitrarily, $X$ must be a \MAXUSW allocation.
\end{proof}

It is worth noting that Algorithm \ref{algo:leximin} can also be used to compute leximin allocations when agents have $\{-1, 0\}\SUB$ and $\{0, 1\}\SUB$ valuations, where leximin allocations are already known to be efficiently computable \citep{barman2023chores, Babaioff2021Dichotomous}.
This is because $\{-1, 0, 1\}$\ONSUB valuations contain the class of $\{-1, 0\}$\SUB and $\{0, 1\}$\SUB valuations (as shown later in Proposition \ref{prop:a2-orderneutral}). 

\section{Properties of the Leximin Allocation}\label{sec:leximin-properties}
Now that we have shown how a leximin allocation can be computed, we explore its connection to other fairness notions. 

 \subsection{Proportionality}

 A popular notion of fairness with mixed goods and chores is that of proportionality. An allocation $X$ is said to be proportional if each agent receives at least an $n$-th fraction of their value for the entire set of items. This is not always possible --- consider an instance with two agents and one high valued item. The fair allocation literature has therefore, instead, studied a relaxation of proportionality called proportionality up to one item~\citep{aziz2022mixedgoodsandchores}. An allocation $X$ is {\em proportional up to one item} (\PROP) if any of the three following conditions hold for every agent $i \in N$:
\begin{enumerate}[(a)]
    \item $v_i(X_i) \ge \frac1n v_i(O)$, 
    \item $v_i(X_i + o) \ge \frac1n v_i(O)$ for some $o \in O \setminus X_i$, or
    \item $v_i(X_i - o) \ge \frac1n v_i(O)$ for some $o \in X_i$.
\end{enumerate}

Leximin allocations are guaranteed to be \PROP. The proof for this requires the following Lemmas.

\begin{lemma}\label{lem:prop-with-one-chore}
When agents have $\{-1, 0, c\}\ONSUB$ valuations, let $X = X^c \cup X^0 \cup X^{-1}$ be a leximin allocation. For any two agents $i$, $j$, if $|X^{-1}_i| > 0$, then $c|X^c_i| - |X^{-1}_i| + 1  \ge c|X^c_j| -|X^{-1}_j| $.
\end{lemma}
\begin{proof}
Assume for contradiction that  $c|X^c_i| - |X^{-1}_i|  < c|X^c_j| -|X^{-1}_j| - 1$ for some agents $i, j \in N$. Construct an allocation $Y^{-1}$ starting at $X^{-1}$ and moving any item $o \in X^{-1}_i$ to $X^{-1}_j$. Define $Y = X^c \cup X^0 \cup Y^{-1}$. 

If $\Delta_{v_j}(X_j, o) \ge 0$, $Y$ Pareto dominates $X$ contradicting the fact that $X$ is leximin. So $\Delta_{v_j}(X_j, o) = -1$. We have 
\begin{align*}
    v_j(Y_j) = c|X^c_j| -|X^{-1}_j| - 1 > c|X^c_i| - |X^{-1}_i| = v_i(X_i).
\end{align*}
We also trivially have $v_i(Y_i) > v_i(X_i)$. Therefore, $\vec s^Y \succ_{\lex} \vec s^X$ contradicting the fact that $X$ is leximin. 
\end{proof}

\begin{lemma}\label{lem:prop-with-no-chore}
When agents have $\{-1, 0, c\}\ONSUB$ valuations, let $X = X^c \cup X^0 \cup X^{-1}$ be a leximin allocation. For any two agents $i$, $j$, if $|X^{-1}_i| = 0$ and there exists an item $o \in X^c_j$ such that $\Delta_{v_i}(X^c_i, o) = c$, then $c|X^c_i| + c \ge c|X^c_j| -|X^{-1}_j| $.
\end{lemma}
\begin{proof}
Assume for contradiction that  $c|X^c_i| + c  < c|X^c_j| -|X^{-1}_j|$ for some agents $i, j \in N$. Construct an allocation $Y^c$ starting at $X^c$ and moving the item $o$ to $X^c_i$. Define $Y = Y^c \cup X^0 \cup X^{-1}$. Using an argument similar to \Cref{thm:exchange-paths}, we can show that $Y^c \cup X^0$ is clean with respect to $\{\beta^0_h\}_{h \in N}$. 

We have the following inequalities for each agent:
\begin{enumerate}[(i)]
    \item $v_h(Y_h) \ge c|Y^c_h| - |X^{-1}_h| = c|X^c_h| - |X^{-1}_h| = v_h(X_h)$ for all $h \in N - i - j$,
    \item $v_i(Y_i) \ge c|Y^c_i| - |X^{-1}_i| = c(|X^c_i|+1) - |X^{-1}_i| > v_i(X_i)$, and
    \item $v_j(Y_j) \ge c|Y^c_j| - |X^{-1}_j| = c(|X^c_j|-1) - |X^{-1}_j| \ge v_j(X_j) - c$.
\end{enumerate}
When comparing $Y$ and $X$, all the agents in $N-i-j$ weakly gain utility, $i$ strictly gains utility and $j$ loses utility by at most $c$. However, note that
\begin{align*}
    v_j(Y_j) \ge c(|X^c_j|-1) - |X^{-1}_j| > c(|X^c_i| + 1) - c \ge v_i(X_i).
\end{align*}

Note that $i$ strictly increases their utility under $Y$ but $j$ loses utility. However,
this still contradicts the fact that $X$ is leximin since $j$ is the only agent to lose utility in $Y$ but still has a higher utility (under $Y$) than $i$ under $X$.
\end{proof}

We are ready to show our main result.

\begin{prop}\label{prop:prop1}
When agents have $\{-1, 0, c\}\ONSUB$ valuations, leximin allocations are \PROP.
\end{prop}
\begin{proof}
Let $X = X^c \cup X^0 \cup X^{-1}$ be any leximin allocation. Fix an agent $i \in N$; we will show that the \PROP conditions hold for each possible $i \in N$. Note that for any item $o' \in X^{-1}_j$ (for some $j \ne i$), $\Delta_i(X_i, o') = -1$ as well. Otherwise, we can move the item $o'$ to $X_i$ to Pareto dominate the allocation $X$. For similar reasons,  for any item $o' \in X^{0}_j$ (for some $j \ne i$), $\Delta_i(X_i, o') \le 0$.

From this, we can conclude that 
\begin{obs}\label{obs:prop1-1}
For any $j \in N - i$, $v_i(X_i \cup X_j) - v_i(X_i) \le c|X^c_j| - |X^{-1}_j|$.
\end{obs}
We break the proof down into cases.

\noindent\textbf{Case 1:} $|X^{-1}_i| > 0$. Let $o$ be some item in $X^{-1}_i$. 

From Lemma \ref{lem:prop-with-one-chore}, we get that
\begin{align*}
    v_i(O) &\le v_i(X_i) + \sum_{j \in N - i} v_i(X_i \cup X_j) - v_i(X_i) \\
    &\le \sum_{h \in N} c|X^c_h| - |X^{-1}_h| \\
    &\le n (c|X^c_i| - |X^{-1}_i| + 1). 
\end{align*}
The first inequality follows from submodularity. The second inequality comes from  Observation \ref{obs:prop1-1}. The final equality comes from Observation \ref{lem:prop-with-one-chore}.

Therefore, $v_i(X_i - o) = c|X^c_i| - |X^{-1}_i| + 1 \ge \frac1n v_i(O)$ and the allocation $X$ is \PROP for the agent $i$.

\noindent\textbf{Case 2:} $|X^{-1}_i| = 0$.
If there is no item in $o \in O \setminus X^c_i$ such that $\Delta_{v_i}(X^c_i, o) = c$. Then the result is trivial since $v_i(X_i) \ge v_i(O) \ge \frac1n v_i(O)$.

Let $N'$ consist of all the agents $j \in N - i$ such that there is some $o' \in X^c_j$ such that $\Delta_{v_i}(X^c_i, o') = c$. We have
\begin{align*}
    v_i(O) &\le v_i(X_i) + \sum_{j \in N'} v_i(X_i \cup X_j) - v_i(X_i) + \sum_{j \in (N\setminus N') - i} v_i(X_i \cup X_j) - v_i(X_i) \\
    &\le c|X^c_i| + \sum_{j \in N'} c|X^c_j| - |X^{-1}_j|  + \sum_{j \in (N\setminus N') - i} (- |X^{-1}_j|)\\
    &\le c|X^c_i| + \sum_{j \in N'} c|X^c_j| - |X^{-1}_j|  + \sum_{j \in (N\setminus N') - i} c|X^c_j|- |X^{-1}_j| \\ 
    &\le n(c|X^c_i| + c).
\end{align*}
The first inequality comes from submodularity. The second inequality comes from the definition of $N'$. The final inequality comes from Lemma \ref{lem:prop-with-no-chore}.

This means, $v_i(X_i + o) = c|X^c_i| + c \ge \frac1n v_i(O)$; we can show that $v_i(X_i + o) = v_i(X_i) + c$ using an argument similar to Lemma \ref{thm:Pareto-improving-paths}. Therefore, $X$ is \PROP for the agent $i$ in this case as well. Since we picked $i$ arbitrarily, we can conclude that $X$ is \PROP for all agents $h \in N$.
\end{proof}

\subsection{Envy-freeness}
Another popular notion of fairness is that of envy-freeness. An allocation is {\em envy-free} if no agent prefers another agent's bundle to their own. This, again, is impossible to guarantee when all items are allocated. 
Due to this impossibility, several relaxations have been studied in the literature. 
The most popular relaxation of envy-freeness is {\em envy-freeness up to one good} (\EFone) \citep{Budish2011EF1,Lipton2004EF1, aziz2022mixedgoodsandchores}. When there are both goods and chores, an allocation $X$ is \EFone if for any two agents $i, j$, there exists $o \in X_i \cup X_j$ such that $v_i(X_i - o) \ge v_i(X_j - o)$.

Unfortunately, leximin allocations are not guaranteed to be \EFone under $\{-1, 0, c\}\ONSUB$ valuations. Consider the following example.

\begin{example}\label{ex:ef1-impossibility}
Consider an instance with two agents $\{1, 2\}$ and six items $G = \{o_1, \dots, o_6\}$. Each agent has valuations defined as follows: 
\begin{align*}
    v_1(S) = (c + 1)[|S \cap \{o_1, o_2\}| + \max\{|S - o_1 - o_2|, 2\}] - |S|, && 
    v_2(S) = (c + 1)|S \cap \{o_1, o_2\}| - |S|.
\end{align*}
It is easy to verify that these two functions are $\{-1, 0, c\}\ONSUB$. In any leximin allocation, agent $2$ gets the items $o_1$, $o_2$ and one item from $\{o_3, o_4, o_5, o_6\}$. Agent $1$ gets the remaining three items. Agent $1$ always envies agent $2$ beyond one item since they value their own bundle at $2c - 1$ but value agent $2$'s bundle at $3c$.
\end{example}

However, under the more restricted class of $\{-1, 0, c\}\ADD$ valuations, leximin allocations are guaranteed to be \EFone.

\begin{prop}\label{prop:add-ef1}
When agents have $\{-1, 0, c\}\ADD$ valuations, leximin allocations are \EFone.
\end{prop}
\begin{proof}
Recall that when agents have additive valuations, each item $o$ has a specific value $v_i(\{o\})$ for the agent $i$. The marginal gain of adding this item to any bundle is $v_i(\{o\})$.

This proof is very similar to \Cref{prop:prop1}. Let $X = X^c \cup X^0 \cup X^{-1}$ be any leximin allocation. Pick any two agents $i, j \in N$. Note that for any item $o' \in X^{-1}_j$, $v_i(\{o'\}) = -1$ as well. Otherwise, we can move the item $o'$ to $X_i$ to Pareto dominate the allocation $X$. For similar reasons,  for any item $o' \in X^{0}_j$, $v_i(\{o'\}) \le 0$.

Therefore, $v_i(X_j) \le c|X^c_j| - |X^{-1}_j|$.

If $|X^{-1}_i| > 0$, let $o$ be some item in $X^{-1}_i$ using Lemma \ref{lem:prop-with-one-chore}, we get $v_i(X_i - o) = c|X^c_i| - |X^{-1}_i| + 1 \ge c|X^c_j| - |X^{-1}_j| \ge v_i(X_j)$. 

If $|X^{-1}_i| = 0$, let $o$ be some item in $X^c_j$ such that $v_i(\{o\}) = c$. If no such item exists, we are done. If such an item exists, from Lemma \ref{lem:prop-with-no-chore}, we get $v_i(X_i) = c|X^c_i| \ge c|X^c_j| - |X^{-1}_j| - c \ge v_i(X_j - o)$.

Since we picked agents $i$ and $j$ arbitrarily, $X$ must be \EFone.
\end{proof}

\subsection{Maxmin Share}
The {\em maxmin share} (\MMS) of an agent $i$ is defined as the value they would obtain had they divided the items into $n$ bundles themselves and picked the worst of these bundles. 
More formally, 
    \begin{align*}
        \MMS_i =  \max_{X = (X_1, X_2, \dots, X_n)} \min_{j\in [n]} v_i(X_j)
    \end{align*} 
\citet{procaccia2014fairenough} show that agents cannot always be guaranteed their maxmin share; past works \citep{Kurokawa2018Maxmin} instead focus on guaranteeing that every agent receives a fraction of their maxmin share. For some $\epsilon \in (0, 1]$, an allocation $X$ is $\epsilon$-$\MMS$ if for every agent $i \in N$, $v_i(X_i) \ge \epsilon\cdot \MMS_i$.

Unfortunately, when agents have $\{-1, 0, c\}\ONSUB$ valuations, there exists an instance where an agent with a positive maxmin share receives a utility of $0$. This rules out the possibility of providing approximate maxmin share guarantees for leximin allocations.

\begin{example}\label{ex:mms-impossibility}
Consider an instance with two agents $\{1, 2\}$ and $10$ items  $\{o_1, \dots, o_{10}\}$. Let $O_{\ell} = \{o_1, \dots, o_l\}$ for any $\ell \in [10]$. Each agent has a valuation function defined as follows:
\begin{align*}
    v_1(S) = \max\{|S \cap O_4|, 2\} + |S \cap (O_6 \setminus O_4)| + |S \cap O_6| - |S|, \\
    v_2(S) = 2|S \cap (O_6 \setminus O_4)| - |S|.
\end{align*}
We can easily verify that these functions are $\{-1, 0, 1\}\ONSUB$. Agent $1$ has a maxmin share of $1$ given by the partition $\{\{o_1, o_2, o_5, o_7, o_8\}, \{o_3, o_4, o_6, o_9, o_{10}\}\}$.

In any leximin allocation, agent $1$ receives the items $\{o_1, o_2, o_3, o_4\}$ and any two out of $\{o_7, o_8, o_9, o_{10}\}$. Agent $2$ receives the other items. This gives agent $1$ a utility of $0$.
\end{example}

However, when agents have $\{-1, 0, c\}\ADD$ valuations, leximin allocations are 1-\MMS. 
The proof for this statement requires some non-trivial but elegant counting arguments similar to the ones used in \citet{cousins2023bivalued}. We find that these counting arguments are best presented as a series of simple Lemmas. 
So for the rest of this subsection, we assume agents have $\{-1, 0, c\}\ADD$ valuations.

Let $X = X^c \cup X^0 \cup X^{-1}$ be an arbitrary leximin allocation. Fix an agent $i \in N$. Let $Y$ be a dominating leximin allocation for the instance where all agents have the valuation function $v_i$ --- no other leximin allocation $Z$ dominates $Y$ when all agents have the valuation function $v_i$ (from \Cref{def:domination}).

We first establish some important properties about $Y$.

\begin{lemma}\label{lem:ycj-bounds}
For all $j \in N$, $|Y^c_n| + 1 \ge |Y^c_j| \ge |Y^c_n|$.
\end{lemma}
\begin{proof}
If there exists a $j \in N$ such that $|Y^c_j| > |Y^c_n| + 1$, then we can use an argument similar to Lemma \ref{lem:xc-equals-yc} (Case 2, sub-case (a)) to show that $Y$ is not a dominating leximin allocation when all agents have the valuation function $v_i$.

If there exists a $j \in N$ such that $|Y^c_j| < |Y^c_n|$, swapping the bundles of $j$ and $n$ i.e. swapping $Y_j$ and $Y_n$, will lead to a leximin allocation $Z$ where $\vec s^{Z^c} = \vec s^{Y^c}$ but $\vec u^{Z^c} \succ_{\lex} \vec u^{Y^c}$. This again, contradicts our choice of $Y$ as a dominating leximin allocation.
\end{proof}

We define $N'$ as the set of agents $j \in N$ where $|Y^c_j| = |Y^c_n| + 1$. Note that for all agents $j \in N \setminus N'$, we must have $|Y^c_j| = |Y^c_n|$ from \Cref{lem:ycj-bounds}. We establish useful bounds for $Y^{-1}_j$ for any $j \in N'$. 

\begin{lemma}\label{lem:yc1-bounds-nprime}
For all $j \in N'$, $|Y^{-1}_n| - 1 \le \max\{|Y^{-1}_j| - c, 0\} \le |Y^{-1}_n|$. 
\end{lemma}
\begin{proof}
The first inequality comes from Lemma \ref{lem:prop-with-one-chore}. We show the second inequality here. 

Assume for contradiction that there is some $j \in N'$ such that $|Y^{-1}_j| > c + |Y^{-1}_n|$. From Lemma \ref{lem:prop-with-one-chore}, we know that $|Y^{-1}_j| \le c + |Y^{-1}_n| + 1$. Therefore, $|Y^{-1}_j| = c + |Y^{-1}_n| + 1$. Define an allocation $Z$ starting at $Y$ and transferring an item $o$ from $Z^{-1}_j$ to $Z^{-1}_n$. Note that $Y$ and $Z$ have the same sorted utility vector but $Z$ lexicographically dominates $Y$ --- contradicting the fact that $Y$ is a dominating leximin allocation. 
\end{proof}
We can also establish similar bounds for any $j \in N \setminus N'$. The proof is omitted due to its similarity with \Cref{lem:yc1-bounds-nprime}.

\begin{lemma}\label{lem:yc1-bounds-rest}
For all $j \in N \setminus N'$, $|Y^{-1}_n| - 1 \le |Y^{-1}_j| \le |Y^{-1}_n|$.
\end{lemma}

Let $\ell^{-1}_i$ denote the number of items in $O$ agent $i$ values at $-1$. This value is well-defined since agents are assumed to have additive valuations. 

\begin{lemma}\label{lem:mms-upperbound}
    $\MMS_i \le v_i(Y_n) = |Y^c_n| - \max\{\lceil \frac{\ell^{-1}_i - c|N'|}{n} \rceil, 0\}$.
\end{lemma}
\begin{proof}
The first inequality holds by the definition of a leximin allocation. This proof focuses on the equality in the statement. To show that $v_i(Y_n) = |Y^c_n| - \max\{\lceil \frac{\ell^{-1}_i - c|N'|}{n} \rceil, 0\}$, it suffices to show that $|Y^{-1}_n| = \max\{\lceil \frac{\ell^{-1}_i - c|N'|}{n} \rceil, 0\}$. 

Note that $\sum_{j \in N} |Y^{-1}_j| = \ell^{-1}_i$. Assume for contradiction that $|Y^{-1}_n| \ne \max\{\lceil \frac{\ell^{-1}_i - c|N'|}{n} \rceil, 0\}$. 

\noindent{\textbf{Case 1:} $|Y^{-1}_n| > \max\{\lceil \frac{\ell^{-1}_i - c|N'|}{n} \rceil, 0\}$}. 
Since $|Y^{-1}_n|$ is an integer, we have $|Y^{-1}_n| \ge \max\{\lceil \frac{\ell^{-1}_i - c|N'|}{n} \rceil, 0\} + 1$. This gives us the following sequence of expressions which reach a contradiction:
\begin{align*}
    \ell^{-1}_i &= \sum_{j \in N} |Y^{-1}_j| \\ 
    &= |Y^{-1}_n| + \sum_{j \in N'} |Y^{-1}_j| +  \sum_{j \in (N \setminus N') - i} |Y^{-1}_j| \\
    &\ge \max\bigg \{\bigg \lceil \frac{\ell^{-1}_i - c|N'|}{n} \bigg \rceil, 0 \bigg \} + 1 + \sum_{j \in N'} \bigg (c+\max\bigg \{\bigg \lceil \frac{\ell^{-1}_i - c|N'|}{n} \bigg \rceil, 0 \bigg \} \bigg ) + \sum_{j \in (N \setminus N') - i} \max\bigg \{\bigg \lceil \frac{\ell^{-1}_i - c|N'|}{n} \bigg \rceil, 0 \bigg \} \\
    &= n\max\bigg \{\bigg \lceil \frac{\ell^{-1}_i - c|N'|}{n} \bigg \rceil, 0 \bigg \} + c|N'| + 1 \\
    &\ge n \bigg (\frac{\ell^{-1}_i - c|N'|}{n} \bigg ) + c|N'| + 1 \\
    &= \ell^{-1}_i + 1 .  
\end{align*}
The first inequality uses both Lemmas \ref{lem:yc1-bounds-nprime} and \ref{lem:yc1-bounds-rest}.

\noindent{\textbf{Case 2:} $|Y^{-1}_n| < \max\{\lceil \frac{\ell^{-1}_i - c|N'|}{n} \rceil, 0\}$}
Since $|Y^{-1}_n|$ is an integer, this can be re-written as $|Y^{-1}_n| < \max\{ \frac{\ell^{-1}_i - c|N'|}{n}, 0\}$. We have another sequence of expressions which reaches a contradiction. 

\begin{align*}
    \ell^{-1}_i &= \sum_{j \in N} |Y^{-1}_j| \\ 
    &= |Y^{-1}_n| + \sum_{j \in N'} |Y^{-1}_j| +  \sum_{j \in (N \setminus N') - i} |Y^{-1}_j| \\
    &< \max \bigg \{\frac{\ell^{-1}_i - c|N'|}{n}, 0 \bigg \} + \sum_{j \in N'} \bigg (c+\max\bigg \{\frac{\ell^{-1}_i - c|N'|}{n}, 0 \bigg \}\bigg ) + \sum_{j \in (N \setminus N') - i} \max\bigg \{\frac{\ell^{-1}_i - c|N'|}{n}, 0 \bigg \} \\
    &= n\max\bigg \{\frac{\ell^{-1}_i - c|N'|}{n}, 0 \bigg \} + c|N'| \\
    &= \max\{\ell^{-1}_i, c|N'|\} . 
\end{align*}

Again, the first inequality uses both Lemmas \ref{lem:yc1-bounds-nprime} and \ref{lem:yc1-bounds-rest}. The only way these inequalities can be true is if $\ell^{-1}_i < c|N'|$. However, in this case, $\max\{\lceil \frac{\ell^{-1}_i - c|N'|}{n} \rceil, 0\}$ becomes $0$ which implies $|Y^{-1}_n|$ is negative --- a contradiction.
\end{proof}

Let us now analyze to the leximin allocation $X = X^c \cup X^0 \cup X^{-1}$. Our first result lower bounds the high valued items that agents have.

\begin{lemma}\label{lem:xcj-lower-bound}
$\sum_{j \in N-i} |X^c_j| \ge |N'| + n|Y^c_n| - |X^c_i|$.
\end{lemma}
\begin{proof}
The agent $i$ has exactly $\sum_{j \in N} |Y^c_j| = n|Y^c_n| + |N'|$ items they value at $c$. This comes from the fact that agents have additive valuations. 

Let us denote the set of all these items by $O'$. When these items are allocated, these items must be allocated in the allocation $X^c$. Otherwise, if an item $o \in O'$ is allocated to some agent $j$ in either $X^0_j$ (or $X^{-1}_j$), moving $o$ from $X^0_j$ to $X^{c}_i$ results in an allocation which Pareto dominates $X$ --- contradicting the fact that $X$ is leximin.

Therefore, $\sum_{j \in N} |X^c_j| \ge n|Y^c_n| + |N'|$ and the Lemma follows.
\end{proof}

Our next Lemma shows a crucial lower bound for $v_i(X_i)$ when $|X^{-1}_i| > 0$.

\begin{lemma}\label{lem:x1i-upperbound}
If $|X^{-1}_i| > 0$, $c|Y^c_n| - \ceil*{ \frac{\ell^{-1}_i - c|N'|}{n} }  \le  c|X^c_i| - |X^{-1}_i|$.

\end{lemma}
\begin{proof}
From Lemma \ref{lem:prop-with-one-chore}, for any agent $j \in N - i$ we have,

\begin{align*}
    |X^{-1}_j| \ge c|X^c_j| - c|X^c_i| + |X^{-1}_i| - 1.
\end{align*}
Summing over all $j \in N-i$ and adding $|X^{-1}_i|$ on both sides we get
\begin{align*}
    \sum_{j \in N} |X^{-1}_j| \ge \sum_{j \in N-i} c|X^c_j| - (n-1)c|X^c_i| + n|X^{-1}_i| - (n-1).
\end{align*}
Plugging in Lemma \ref{lem:xcj-lower-bound}, we get
\begin{align*}
    \sum_{j \in N} |X^{-1}_j| \ge c|N'| + nc|Y^c_n| - nc|X^c_i| + n|X^{-1}_i| - n + 1.
\end{align*}
Note that $\sum_{j \in N} |X^{-1}_j| \le \ell^{-1}_i$. If this is not true, there must be some item $o'$ which agent $i$ values non-negatively but $o'$ is allocated to an agent $j$ who values it at $-1$. Transferring $o'$ from $j$ to $i$ will result in a Pareto improvement --- contradicting the fact that $X$ is leximin. Therefore, we have
\begin{align*}
    \ell^{-1}_i \ge c|N'| + n \big [c|Y^c_n| - c|X^c_i| + |X^{-1}_i| - 1 \big ] + 1.
\end{align*}
Re-arranging this, 
\begin{align*}
     c|X^c_i| - |X^{-1}_i|   \ge  c|Y^c_n| - \bigg (\frac{\ell^{-1}_i - c|N'| - 1}{n} + 1\bigg ).
\end{align*}
Since $c|X^c_i| - |X^{-1}_i|$ is an integer, we can replace $\frac{\ell^{-1}_i - c|N'| - 1}{n}$ with its floor value.
\begin{align*}
     c|X^c_i| - |X^{-1}_i|   \ge  c|Y^c_n| - \bigg (\bigg \lfloor\frac{\ell^{-1}_i - c|N'| - 1}{n} \bigg \rfloor + 1\bigg ).
\end{align*}
With a little algebra, we can show that $\lfloor\frac{\ell^{-1}_i - c|N'| - 1}{n} \rfloor + 1 \le \lceil \frac{\ell^{-1}_i - c|N'|}{n} \rceil$. This completes the proof.
\end{proof}

Our next result takes care of the case where $|X^{-1}_i| = 0$. Before we present it, we require the following definition.

\begin{definition}[Balanced]
An allocation is said to be {\em balanced} if for all $i, j \in N$, if $|X^c_i| + 2 \le |X^c_j|$, then for all $o \in X^c_j$ $v_i(\{o\}) \le 0$.
\end{definition}
The balanced property is very useful when analyzing allocations since it makes counting the number of high valued items easier.
All leximin allocations are not balanced but for every leximin allocation there exists a balanced leximin allocation with the same utility vector. 

\begin{lemma}\label{lem:leximin-balanced}
When agents have $\{-1, 0, c\}\ADD$ valuations, for any leximin allocation $X = X^c \cup X^0 \cup X^{-1}$, there exists a balanced leximin allocation $Z = Z^c \cup Z^0 \cup Z^{-1}$ with the same utility vector as $X$.
\end{lemma}
\begin{proof}
We use a similar argument to Lemma \ref{lem:xc-equals-yc}. Assume there exists two agents $i, j \in N$ such that $|X^c_i| + 2 \le |X^c_j|$ and $v_i(\{o\}) = c$ for some $o \in X^c_j$. 

Construct the allocation $Z^c$ starting at $X^c$ and moving $o$ from $Z^c_j$ to $Z^c_i$. 
Note that $\sum_{h \in N} |X^c_h|^2 > \sum_{h \in N} |Z^c_h|^2$.
We have two cases. 

If $|X^{-1}_j| < c$, define the allocation $Z$ as $Z^c \cup X^0 \cup X^{-1}$. We have the following equations for each agent $h \in N$:
\begin{enumerate}[(i)]
    \item $v_h(Z_h) = c|Z^c_h| - |X^{-1}_h| = c|X^c_h| - |X^{-1}_h| = v_h(X_h)$ for all $h \in N - i - j$,
    \item $v_i(Z_i) = c|Z^c_i| - |X^{-1}_i| = c(|X^c_i|+1) - |X^{-1}_i| > v_i(X_i)$, and
    \item $v_j(Z_j) = c|Z^c_j| - |X^{-1}_j| = c(|X^c_j|-1) - |X^{-1}_j| = v_j(X_j) - c$.
\end{enumerate}
When comparing $Z$ and $X$, all the agents in $N-i-j$ weakly gain utility, $i$ strictly gains utility and $j$ loses utility by $c$. However, note that

$$v_j(Z_j) = c(|X^c_j|-1) - |X^{-1}_j| \ge c(|X^c_j| + 1) - |X^{-1}_j| > c(|X^c_i| + 1) - c \ge v_i(X_i).$$

This contradicts the fact that $X$ is leximin since $j$ is the only agent to lose utility in $Z$ but still has a higher utility than $i$ under $X$.

Therefore, it must be the case that $|X^{-1}_j| \ge c$. Let $S$ be a $c$-sized subset of $X^{-1}_j$. Consider the allocation $Z^c \cup C^0 \cup Z^{-1}$ where $Z^{-1}$ is defined as follows:
\begin{align*}
    Z^{-1}_k = 
    \begin{cases}
        X^{-1}_k & (k \in N - i - j) \\
        X^{-1}_j \setminus S & (k = j) \\
        X^{-1}_i \cup S & (k = i)
    \end{cases}
    .
\end{align*}
Define $Z = Z^c \cup X^0 \cup Z^{-1}$.
We can set up a similar set of equations for each agent $h \in N$.
\begin{enumerate}[(i)]
    \item $v_h(Z^c_h) = c|Z^c_h| - |Z^{-1}_h| = c|X^c_h| - |X^{-1}_h| = v_h(X_h)$ for all $h \in N - i - j$,
    \item $v_i(Z_i) = c|Z^c_i| - |Z^{-1}_i| = c(|X^c_i|+1) - |X^{-1}_i| - c = v_i(X_i)$, and
    \item $v_j(Z_j) \ge c|Z^c_j| - |Z^{-1}_j| = c(|X^c_j|-1) - |X^{-1}_j| + c = v_j(X_j)$.
\end{enumerate}
$Z$ and $X$ have the same utility vector, implying that they are both leximin. 

However, $Z$ still may not satisfy the requirement of this lemma. What we have shown is that whenever there exists an allocation $X$ which does not satisfy the requirement, we can make simple transfers to create another allocation $Z$ with the same utility vector such that $\sum_{h \in N} |X^c_h|^2 > \sum_{h \in N} |Z^c_h|^2$. Since $\sum_{h \in N} |X^c_h|^2$ is a finite integer lower bounded at $0$, we can repeatedly make the transfers described in this proof to end with an allocation $Z$ that satisfies the requirements of the Lemma.
\end{proof}

We are now ready to prove our main result.

\begin{prop}
When agents have $\{-1, 0, c\}\ADD$ valuations, let $X$ be any leximin allocation. For all agents $i \in N$, $v_i(X_i) \ge \MMS_i$.
\end{prop}
\begin{proof}
Let $Z= Z^c \cup Z^0 \cup Z^{-1}$ be a balanced leximin allocation with the same utility vector as $X$. If we show that $Z$ guarantees each agent their maxmin share, then it implies that $X$ guarantees each agent their maxmin share as well since they have the same utility vector.

Fix an agent $i \in N$. Let $Y = Y^c \cup Y^0 \cup Y^{-1}$ be a dominating leximin allocation where all agents have the valuation function $v_i$. We have two cases.

\noindent\textbf{Case 1:} $|Z^{-1}_i| > 0$. The proof for this case is almost obvious from \Cref{lem:x1i-upperbound}. We have 
\begin{align*}
    \MMS_i \le c|Y^c_n| - \max\bigg \{\bigg \lceil \frac{\ell^{-1}_i - c|N'|}{n} \bigg \rceil, 0\bigg \} \le c|Y^c_n| - \bigg \lceil \frac{\ell^{-1}_i - c|N'|}{n} \bigg \rceil \le c|Z^c_i| - |Z^{-1}_i| = v_i(Z_i).
\end{align*}
The first equality follows from Lemma \ref{lem:mms-upperbound}. The final inequality follows from Lemma \ref{lem:x1i-upperbound}.

\noindent\textbf{Case 2:} $|Z^{-1}_i| = 0$. In this case, we need to show that $c|Z^c_i| \ge  c|Y^c_n| - \max \{ \lceil \frac{\ell^{-1}_i - c|N'|}{n}  \rceil, 0\}$.

The number of items that agent $i$ values at $c$ is equal to $n|Y^c_n| + |N'|$. If $|Z^c_i| < |Y^c_n|$, using the pigeonhole principle, there must be some agent $j \in N$ with at least $|Z^c_i| + 2$ of these items in the allocation $Z$. All of these $|Z^c_i| + 2$ items must be in $|Z^c_j|$ since otherwise, we can make simple transfers to Pareto dominate the allocation $X$. This contradicts the assumption that $Z$ is balanced. 

Therefore, it must be the case that $|Z^c_i| \ge |Y^c_n|$. Note that this trivially implies that $v_i(Z_i) \ge \MMS_i$
\begin{align*}
    \MMS_i \le c|Y^c_n| - \bigg \lceil \frac{\ell^{-1}_i - c|N'|}{n} \bigg \rceil \le c|Y^c_n| \le c|Z^c_i| = v_i(Z_i).
\end{align*}
\end{proof}

\subsection{Lorenz Dominance and Nash Welfare}
An allocation $X$ is {\em Lorenz dominating} \citep{Babaioff2021Dichotomous} if for all allocations $Y$ and all $k \in [n]$, it holds that $\sum_{j \in [k]} \vec s^X_j \ge \sum_{j \in [k]} \vec s^Y_j$ where $\vec s^X$ is the sorted utility vector of the allocation $X$ (defined in Section \ref{sec:fairness-objectives}).



We show that all leximin allocations are Lorenz dominating when agents have $\{-1, 0, c\}\ONSUB$ valuations.

\begin{prop}
When agents have $\{-1, 0, c\}\ONSUB$ valuations, all leximin allocations are Lorenz dominating.
\end{prop}
\begin{proof}
Since all leximin allocations have the same sorted utility vector, showing that one leximin allocation is Lorenz dominating is equivalent to showing all leximin allocations are Lorenz dominating. Let $X = X^c \cup X^0 \cup X^{-1}$ be the leximin allocation computed by Algorithm \ref{algo:leximin}.

Assume for contradiction that $X$ is not Lorenz dominating. We will say that allocation $Y$ {\em blocks} $X$ from being Lorenz dominating if there exists a $k \in [n]$ such that $\sum_{j \in [k]} \vec s^Y_j > \sum_{j \in [k]} \vec s^X_j$.
If $X$ is not Lorenz dominating, there must be at least one blocking allocation $Y$. 
If there are multiple such blocking allocations $Y$, choose $Y = Y^c \cup Y^0 \cup Y^{-1}$ such that $\sum_{h \in N} \abs{|Y^c_h| - |X^c_h|}$ is minimized; further break ties by choosing $Y$ such that $\sum_{h \in N} \abs{|Y^{-1}_h| - |X^{-1}_h|}$ is minimized. 

If for any $i \in N$, $|X^c_i| \ne |Y^c_i|$, we can use the same analysis from Lemma \ref{lem:xc-equals-yc} to construct an allocation $\hat Y = \hat{Y}^c \cup \hat{Y}^0 \cup \hat{Y}^{-1}$ such that $\hat Y$ blocks $X$ and $\sum_{h \in N} \abs{|Y^c_h| - |X^c_h|} > \sum_{h \in N} \abs{|\hat{Y}^c_h| - |X^c_h|}$ --- contradicting our choice of $Y$.

Therefore, for all $i \in N$, we can assume $|X^c_i| = |Y^c_i|$. If for all $i \in N$, $|X^{-1}_i| \le |Y^{-1}_i|$ as well, we are done since there is no way $Y$ can block $X$. 

So assume $|X^{-1}_i| > |Y^{-1}_i|$ for some $i \in N$. Since we assumed $X$ was computed using Algorithm \ref{algo:leximin}, we know that $\sum_{h \in N} |X^{-1}_h| \le \sum_{h \in N} |Y^{-1}_h|$. Therefore, there must exist some $j \in N$ such that $|X^{-1}_j| < |Y^{-1}_j|$. We can now use analysis from \Cref{thm:leximin} to construct another allocation $\hat Y = \hat{Y}^c \cup \hat{Y}^0 \cup \hat{Y}^{-1}$ such that $\hat Y$ blocks $X$ and the following two conditions hold:
\begin{enumerate}[(i)]
    \item $\sum_{h \in N} \abs{|Y^c_h| - |X^c_h|} = \sum_{h \in N} \abs{|\hat{Y}^c_h| - |X^c_h|}$, and
    \item $\sum_{h \in N} \abs{|Y^{-1}_h| - |X^{-1}_h|} > \sum_{h \in N} \abs{|\hat{Y}^{-1}_h| - |X^{-1}_h|}$.
\end{enumerate}
This again contradicts our choice of $Y$. From this we can conclude that a blocking allocation $Y$ must not exist, and therefore $X$ is Lorenz dominating.
\end{proof}

Lorenz dominance also implies several other fairness notions as shown in the following Theorem. Similar results are well-known for the specific case of non-negative utilities (see for example, \citet{arnold1987lorenz}). We generalize it to real valued utilities, so we can show an interesting connection to other well-known fairness objectives like $p$-mean welfare and Nash welfare.

\begin{theorem}[Lorenz Dominance and Welfare Optimality]
\label{thm:lorenz-conseqences}
Suppose that an allocation $X$ is Lorenz-dominant w.r.t.\
some class of allocations. 
Given a concave monotonically-increasing function $f: \R \to \R$, define the 
symmetric welfare concept $\Welfare_{f}(\vec{u}^{X}) \doteq \sum_{i=1}^{N} f({u}^{X}_{i})$.
The following hold.
\begin{enumerate}
\item \label{thm:lorenz-conseqences:ul2w} $X$ is a $\Welfare_{f}$ maximizing solution. 

\item \label{thm:lorenz-conseqences:uw2l} If $X'$ is a $\Welfare_{f}$ maximizing solution, and $f$ is strictly concave, then $X'$ is Lorenz-dominant.

\end{enumerate}
Furthermore, given a concave monotonically-increasing function $g: \R_{0+} \to \R \cup \{-\infty\}$, define the 
welfare concept $\Welfare_{g}(\vec{u}^{X}) \doteq \sum_{i=1}^{N} g({u}^{X}_{i})$.
The following then hold.
\begin{enumerate}[resume]
\item \label{thm:lorenz-conseqences:bexist} 
A $\Welfare_{g}$ maximizing solution exists iff $s^X_1 \geq 0$ (i.e., iff the welfare of any solution is defined).

\item \label{thm:lorenz-conseqences:bl2w} If a $\Welfare_{g}$ maximizing solution exists, then $X$ is a $\Welfare_{g}$ maximizing solution.

\item \label{thm:lorenz-conseqences:bw2l} If $X'$ is a $\Welfare_{g}$ maximizing solution, and $g$ is strictly concave, then $X'$ is Lorenz-dominant.

\item \label{thm:lorenz-conseqences:bnl2w} If $j$ is the index of the first non-negative entry of the Lorenz-dominant sorted utility vector $\vec{s}^{X}$, then among all feasible $Y$ such that $\forall i < j$: $\vec{s}^{Y}_{i} \geq \vec{s}^{X}_{i}$ [alternatively, $\vec{s}^{Y}_{1:j-1} \succeq \vec{s}^{X}_{1:j-1}$], 
$X$ maximizes $\sum_{ \{i \in N | \vec{s}^{X}_{i} \geq 0 \} } g(\vec{s}^{X}_{i})$.
\end{enumerate}
%
%
\end{theorem}
\begin{proof}
Before handling each case, we note that the key to this result is that any feasible utility vector can be constructed from the Lorenz-dominant utility vector $\vec{u}^{X}$ by iteratively applying two operations, namely the \emph{inequitable transfer} operation and the \emph{utility loss} operation (henceforth formalized), and both operations can only decrease the relevant welfare concepts.
This immediately yields \cref{thm:lorenz-conseqences:ul2w}, and \cref{thm:lorenz-conseqences:uw2l} follows via a strictness argument.
From here, \cref{thm:lorenz-conseqences:bexist} is a straightforward independent argument, and once existence is assumed, \cref{thm:lorenz-conseqences:bl2w,thm:lorenz-conseqences:bw2l} follow by similar reasoning to \cref{thm:lorenz-conseqences:ul2w,thm:lorenz-conseqences:uw2l}. 
Finally, \cref{thm:lorenz-conseqences:bnl2w} furthers this theme by making similar arguments after constraining the feasible class of allocations to essentially be Lorenz-dominant over the negative utility portion of the allocation.

We first argue that if $X$ is Lorenz-dominant with utility vector $\vec{u}$, then any utiilty vector $\vec{u}'$ that arises from some feasible $X'$ can be constructed from $\vec{u}$ using the following two operations.
For convenience, 
operate henceforth on sorted utility vectors, as Lorenz-dominance, $\Welfare_{f}(\vec{u})$, and $\Welfare_{f}$ optimality all exhibit symmetry over utility values, and can thus be computed from sorted utility vectors. 
Let $\vec s$ be the sorted utility vector of the Lorenz dominating allocation $\vec {s'}$ be any other feasible sorted utility vector.
\newcommand{\1}{\mathbbm{1}}

We demonstrate the construction of $\vec{s'}$ from $\vec{s}$ by iterating the following two operations. In the definition of both these operations $\1_{i}$ denotes an indicator vector.

\begin{enumerate}
\item \textbf{Inequitable Transfer}: $\vec{s} \mapsto \vec{s} - \varepsilon \1_{i} + \varepsilon \1_{j}$ for some $\varepsilon > 0$ and (sorted) agent indices
$i, j \in [n]$ s.t.\ $i < j$,
and by sortedness, ${s}_{i} \leq {s}_j$, i.e., $\varepsilon$ utility is transferred from an agent with $i$th smallest to $j$th smallest utility. 
\item \textbf{Utility Loss}: $\vec{s} \mapsto \vec{s} - \varepsilon \1_{i}$ for some $\varepsilon > 0$ and sorted agent index $i \in [n]$.
\end{enumerate}
Essentially, \emph{inequitable transfer} represents an inequitable redistribution of utility; a perverse theft by the rich from the poor, and the \emph{utility loss} operation simply characterizes a loss of utility by some agent.
Intuitively, these operations sacrifice Lorenz-dominance, as neither operation can produce 
an $\vec{s'}$ that Lorenz-dominates
$\vec{s}$, 
and their opposites (equitable transfer from high to low utility, or gain in utility) could be used to construct a $\vec{s'}$ that Lorenz-dominates $\vec{s}$.
Furthermore, it is straightforward to see that if $\vec{s}'$ is Lorenz-dominated by $\vec{s}$, then a sequence of \emph{inequitable transfer} operations can be used to produce some $\vec{s''}$ such that $\vec{s''} \succeq \vec{s'}$, and then one or more \emph{utility loss} operations produce $\vec{s'}$ from $\vec{s''}$.
Now, since we assumed that $\vec{s}$ Lorenz-dominates all feasible allocations, and we have just shown that any allocation $\vec{s'}$ Lorenz-dominated by $\vec{s}$ can be constructed as above, it holds that all feasible allocation utility vectors $\vec{s'}$ can be constructed from $\vec{s}$.

\medskip

We are now ready to show \cref{thm:lorenz-conseqences:ul2w,thm:lorenz-conseqences:uw2l}.
The key step is to show that any utility vector $\vec{u}'$ that arises from some feasible $X'$ exhibits $\Welfare_{f}(\vec{u}') \leq \Welfare_{f}(\vec{u})$.
Observe that, as above, $\vec{u}$ can be constructed from $\vec{u}$ by applying the \emph{inequitable transfer} and \emph{utility loss} operations, and moreover, by concavity, of $f$, {inequitable transfer} can only decrease $\Welfare_{f}$, and by monotonic increasingness, {utility loss} can only decrease $\Welfare_{f}$, thus we may conclude that $\Welfare_{f}(\vec{u}^{X}) \geq \Welfare_{f}(\vec{u}^{X'})$ for all feasible allocations $X'$.
This immediately yields \cref{thm:lorenz-conseqences:ul2w}.
Now, observe that if concavity is strengthened to strict concavity, then the welfare loss of both the {inequitable transfer} and {utility loss} operations is \emph{strictly positive}, therefore $X$ \emph{uniquely maximizes} $\Welfare_{f}$, and \cref{thm:lorenz-conseqences:ul2w,thm:lorenz-conseqences:uw2l} 
 follows immediately. 

\medskip

We now show \cref{thm:lorenz-conseqences:bexist}.
Recall that
we now operate w.r.t.\ some $g: \R_{0+} \to \R$ that is undefined for negative inputs.
The task can be restated as
\[
\text{A $\Welfare_{g}$ maximizing solution exists} \Leftrightarrow s^X_1 \geq 0
,
\]
and the forward direction follows straightforwardly via contraposition,
while the reverse direction
follows directly via a subtler argument.
In particular, if the minimum utility of the Lorenz-dominant allocation is negative, i.e., $s^X_1 < 0$, then there do not exist fully non-negative allocations, and $\Welfare_{g}(\vec{s}^{X'})$ is undefined for all feasible $X'$.
For the reverse direction, observe that $\Welfare(\vec{s}^{X})$ is at least well-defined,
and furthermore,
\cref{thm:lorenz-conseqences:bl2w}
(shown subsequently) demonstrates that the $\vec{s}^{X}$ is indeed the maximizer. 
Moreover, we require \cref{thm:lorenz-conseqences:bl2w}
to show only the reverse direction, whereas the proof of \cref{thm:lorenz-conseqences:bl2w} depends only on the forward direction, thus there is no cyclic dependence in this reasoning.


\medskip

We now show the remaining items.
Observe that \cref{thm:lorenz-conseqences:bl2w,thm:lorenz-conseqences:bw2l} follow by the reasoning of \cref{thm:lorenz-conseqences:ul2w,thm:lorenz-conseqences:uw2l}, as by \cref{thm:lorenz-conseqences:bexist}, if a Lorenz-dominating allocation exists, then $s^X_1 \geq 0$, which by definition implies $\vec{s}^X \succeq \bm{0}$, thus $\Welfare_{g}(\vec{s}^{X})$ is well-defined.
The only difference between this case and \cref{thm:lorenz-conseqences:bl2w,thm:lorenz-conseqences:bw2l} is that we now need to consider only the subset feasible allocations with non-negative utility vectors, but observe that this subset can also be traversed with the inequitable transfer and utility loss operators (as neither operator can convert a utility vector $\vec{s}'$ with $s_{1} < 0$ to some $\vec{s}''$ with $s_{1} \geq 0$).
Finally, \cref{thm:lorenz-conseqences:bnl2w} follows via similar reasoning, now observing that the claim reduces to \cref{thm:lorenz-conseqences:ul2w} after constraining the set of feasible allocation to those that match the first $i$ values of $\vec{s}^{X}$ (i.e., are equivalent on the necessary negative valuations), further constraining to $X'$ such that ${s}^{X'}_{j}$ is non-negative (as otherwise the objective is undefined), and then removing agents with negative valuations and applying \cref{thm:lorenz-conseqences:ul2w} to the remaining class.
\end{proof}

Observe that \cref{thm:lorenz-conseqences:ul2w,thm:lorenz-conseqences:uw2l} apply to general symmetric welfare functions of unbounded utility (e.g., utilitarian, egalitarian, generalized Gini social welfare functions, et cetera), whereas \cref{thm:lorenz-conseqences}~\cref{thm:lorenz-conseqences:bl2w,thm:lorenz-conseqences:bw2l} apply to symmetric welfare functions of
non-negative
utility.
The Nash social welfare $\Welfare_{0}$ and the power-mean family $\Welfare_{p}$, defined for $p \leq 1$ as
\begin{equation}
\label{eq:nash-pm}
\Welfare_{p}(X) \doteq \sqrt[p]{\frac{1}{n} \sum_{i=1}^{n} v_{i}^{p}} 
    \quad \text{or} \quad
    \Welfare_{0}(X) \doteq \exp\left(\frac{1}{n} \sum_{i=1}^{n} \ln(v_{i}) \right),
\end{equation}
exemplify the latter case, assuming we handle $\frac{1}{0}$ and $\ln(0)$ as their right-limit $-\infty$ for $p \leq 0$.
The $p$-mean welfare functions have been extensively studied as fairness objectives in economics \citep{moulin2004fair}, fair machine learning \citep{heidari2018fairness,cousins2021bounds,cousins2021axiomatic,cousins2022uncertainty}, and, more recently, fair allocation \citep{barman2020tight}.
When $p$ approaches $-\infty$, the $p$-mean welfare corresponds to the leximin objective, and of course the Nash social welfare ($p=0$) and utilitarian welfare ($p=1$) are also included as special cases.


In the only-goods setting, the above welfare functions are always defined, thus maximization over finite classes is always possible, but in the mixed manna setting, for welfare functions defined only for non-negative utilities, the relationship between Lorenz-dominance and $\Welfare_{f}$ optimality requires additional qualifiers (\cref{thm:lorenz-conseqences}~\cref{thm:lorenz-conseqences:bexist}).
Essentially, so long as it is possible for all agents to simultaneously receive positive-utility bundles, non-negative welfare concepts are well-defined, and the Lorenz-dominance of \cref{algo:leximin} implies that we recover these optimal solutions. 
Furthermore, \cref{thm:lorenz-conseqences:bnl2w} tells us that if negative utilities are unavoidable, the negative portion of the allocation is Lorenz-dominant, and the welfare of the non-negative portion is welfare-optimal.

\section{NP-Hardness when $c$ is not an Integer}\label{sec:lower-bounds}
A natural question to ask given the positive results from the previous sections is whether these results can be generalized to broader classes of valuations. Specifically, we may want to generalize our results from $\{-1, 0, c\}\ONSUB$ valuations to $\{-p, 0, q\}\ONSUB$ valuations for arbitrary integers $p$ and $q$.

In this section we show that the problem of computing leximin allocations is NP-hard even for $\{-p, q\}\ADD$ valuations for any co-prime integers $p$ and $q$ such that $p \ge 3$. This proof is very similar to the hardness result in \citet{akrami2022mnw}. Note that the assumption that $p$ and $q$ are co-prime is necessary since if $p$ divides $q$, the problem reduces to computing a leximin allocation for agents with $\{-1, \frac{q}{p}\}\ADD$ valuations and admits a polynomial time algorithm (Theorem \ref{thm:leximin}). More generally, any common divisor of $p$ and $q$ can be eliminated by scaling agent valuations. 

\begin{theorem}\label{thm:lower-bound}
The problem of computing leximin allocations is NP-hard even when agents have $\{-p, q\}\ADD$ valuations for any co-prime positive integers $p$ and $q$ such that $p \ge 3$.  
\end{theorem}
\begin{proof}
We present a reduction to the Exact $p$-dimensional matching problem (Ex-$p$-DM) defined as follows: given a graph $G$ with $p$ sets of vertices $V_1$, $V_2$, \dots, $V_p$ each of size $a$ and a set of edges $E \subseteq V_1 \times V_2 \times \dots \times V_p$ of size $b$, the goal is to decide if there exists a $p$-dimensional perfect matching in $G$. This problem is a generalization of the Exact $3$-dimensional matching problem and was shown to be NP-complete by \citet{akrami2022mnw}.

Given an Ex-$p$-DM instance $G$, we construct a fair allocation instance as follows: each vertex in the graph $G$ is an item and we add $aq$ dummy items. Each edge in the graph $G$ is an agent who values the items corresponding to the nodes incident on it at $q$ and all the other items at $-p$.

If there is a perfect matching in $G$, then there is an allocation where each agent receives a utility of at least $0$. Pick an arbitrary perfect matching. We can construct an allocation as follows: the agents corresponding to the edges in the perfect matching receive the $p$ items incident on it (in $G$) as well as any $q$ of the dummy items. All other agents receive nothing. Note that this is a complete allocation since there are exactly $a$ agents in the perfect matching. It is easy to verify that all agents receive a utility of exactly $0$. All the agents corresponding to edges outside the perfect matching receive no items and so receive a utility of $0$. All the agents corresponding to edges in the perfect matching receive $p$ items they value at $q$ and $q$ items they value at $-p$.

If we show that at least one agent receives a negative utility in all allocations when there is no perfect matching, we are done. This is because it implies that any leximin allocation gives at least one agent a negative utility when there is no perfect matching but when there is a perfect matching, all agents are given a non-negative utility in any leximin allocation.

Note that in our constructed fair allocation instance, the maximum \USW that can be achieved is $0$ --- there are $aq$ items all agents value at $-p$ and $ap$ items some agent values at $q$. 

Therefore, if in any allocation $X$, if some agent receives a positive utility then some other agent must receive a negative utility and we are done. So, we only need to show that it is not possible for all agents to receive a utility of exactly $0$. Since $p$ and $q$ are co-prime and each agent values exactly $p$ items at $q$, an agent $i$ receives a utility of $0$ if and only if one of the following conditions is satisfied:
\begin{enumerate}[(i)]
    \item $i$ receives an empty bundle, or
    \item $i$ receives the $p$ items they value at $q$ and receives $q$ additional items they value at $p$.
\end{enumerate}
So, if all agents receive a utility of $0$ in some allocation $X$, the edges corresponding to the agents who received the $p$ items they value at $q$ form a perfect matching in $G$. Since we assumed there is no perfect matching in $G$, it cannot be the case that all agents receive a utility of $0$. Therefore, some agent must receive negative utility in any allocation $X$ when there is no perfect matching. This completes the proof.
\end{proof}

While Theorem \ref{thm:lower-bound} shows that the problem of computing leximin allocations is NP-hard for most values of $p$ and $q$, there are two special cases which still remain unresloved ---  $\{-c, 0, 1\}\ONSUB$ valuations and $\{-2, 0, c\}\ONSUB$ valuations. Both these classes are outside the scope of this paper but are interesting questions for future work. 

\section{Discussion}\label{sec:discussion}
While order-neutrality is certainly desirable from a computational perspective, it is not immediately clear what kind of non-additive preferences it can capture. We present some examples in this section and also discuss how order-neutral submodular valuations relate to the class of OXS valuations --- a popular subclass of submodular valuations previously studied in fair allocation \citep{benabbou2021MRF}.

Our first observation is that the class of $\{-1, 0, c\}\ONSUB$ valuations contains the class of $\{0, c\}\SUB$ valuations as well as the class of $\{-1, c\}\SUB$ valuations. This is given by the following Proposition.

\begin{prop}\label{prop:a2-orderneutral}
When $|A| = 2$, any $A$\SUB function is order-neutral.  
\end{prop}
\begin{proof}
Let $v$ be an arbitrary $A$\SUB valuation. Since $|A| = 2$, we refer to $v$ as an $\{a, b\}$\SUB valuation for some $a, b \in \R$ such that $a < b$. 

Assume for contradiction that $v$ is not order-neutral. That is, for some bundle $S$, there exists two orderings of items in the bundle $\pi, \pi': [|S|] \rightarrow S$ such that $\vec v(S, \pi) \ne \vec v(S, \pi')$. Since the two vectors only have two distinct values $a$ and $b$, let $\vec v(S, \pi)$ consist of $k$ elements with value $a$ and $|S| - k$ elements with value $b$. Similarly, let $\vec v(S, \pi')$ consist of $k'$ elements with value $a$ and $|S| - k$ elements with value $b$. If $k \ne k'$, then the sums of the two vectors will be different. However, since both the vectors sum up to $v(S)$, we must have $k = k'$. This implies, $\vec v(S, \pi) = \vec v(S, \pi')$ --- a contradiction.
\end{proof}

We next present an example showing that $\{-1, 0, c\}\ONSUB$ valuations can be used to express cardinality constraints over items.
\begin{example}
Consider an agent who values some items $C$ at $c$, some items $Z$ at $0$ and the remaining items at $-1$. This naturally specificies an additive valuation. However, in the presence of decreasing marginal gains, the valuation is no longer additive --- assume that after receiving $k$ items in $C$, adding another item from $C$ gives the agent a marginal utility of $-1$ (as opposed to $c$). Alternatively, assume that after receiving $k'$ items in $Z$, adding another item from $Z$ gives the agent a marginal utility of $-1$. 
These natural classes of preferences cannot be expressed using additive valuations but can be expressed using $\{-1, 0, c\}\ONSUB$ valuations. Examples of this can be found in Examples \ref{ex:ef1-impossibility} and \ref{ex:mms-impossibility}.
\end{example}

Finally, we discuss the relationship between order-neutral submodular functions and OXS functions. To define the class of OXS valuations, we must first define the simpler class of unit demand valuations. A valuation function $f: 2^{O} \rightarrow \R$ is a {\em unit demand valaution} if the value of a set of items is equal to the value of the most valuable item in this set. More formally, for all $S \subseteq O$, $f(S) = \max_{o \in S} f(\{o\})$. 
\begin{definition}[OXS valuations \citep{lehmann2001oxs}]
A function $v:2^O \rightarrow \R$ is OXS if there exists $k$ unit demand functions $f_1, \dots, f_k$ such that for any $S \subseteq N$
\begin{align*}
    v_i(S) = \max\{f_1(S_1) + f_2(S_2) + \dots f_k(S_k) | \{S_1, \dots, S_k\} \text{ is a $k$-partition of }S\}
\end{align*}
Given a set of real numbers $A$, a function $v$ is said to be $A\OXS$ if it satifies the condition above and for all $S \subseteq O$ and $o \in O \setminus S$, $\Delta_{v}(S, o) \in A$.
\end{definition}
The class of OXS valuations contains the class of additive valuations but is strictly contained by the class of submodular valuations \citep{lehmann2001oxs}. We show that OXS valuations and order-neutral submodular valuations are incomparable --- that is, neither class contains the other. This can be seen using the following two Propositions.

\begin{prop}
    $\{-1, 0, 1\}\OXS \not\subset \text{ONSUB}$  
\end{prop}
\begin{proof}
Consider an OXS valuation over two items $\{o_1, o_2\}$ defined by two unit demand functions $f_1, f_2$ as follows
\begin{align*}
    f_1(\{o_1\}) = 1 && f_1(\{o_2\}) = 0 \\ 
    f_2(\{o_1\}) = 0 && f_2(\{o_2\}) = -1
\end{align*}
Note that $v(\{o_1, o_2\}) = 0$ but the sorted telescoping sum vector can be different based on the ordering. If we add $o_1$ first, the sorted telescoping sum vector is $(-1, 1)$. However, if we add $o_2$ first, the sorted telescoping sum vector is $(0, 0)$. Therefore, this function is not order-neutral.
\end{proof}

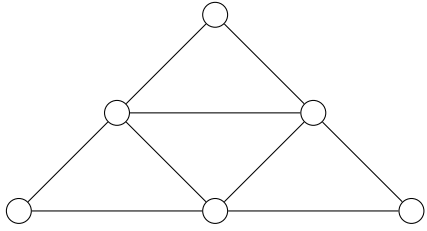
\begin{figure}[t!]
    \centering
    \begin{tikzpicture}[node distance = 1.5cm, radius=0.5pt]
    \node[circle, draw = black] (n1) {};
    \node[circle, draw = black, above right=of n1] (n2) {};
    \node[circle, draw = black, below right=of n2] (n3) {};
    \node[circle, draw = black, above right=of n3] (n4) {};
    \node[circle, draw = black, below right=of n4] (n5) {};
    \node[circle, draw = black, above left=of n4] (n6) {};
    
    \draw (n1) -- (n2) node [midway, fill=white, text opacity=1, opacity=0, left=2pt] {$e_1$};
    \draw (n1) -- (n3) node [midway, fill=white, text opacity=1, opacity=0, below=2pt] {$e_2$};
    \draw (n2) -- (n3) node [midway, fill=white, text opacity=1, opacity=0, left=2pt] {$e_3$};
    \draw (n4) -- (n2) node [midway, fill=white, text opacity=1, opacity=0, above=2pt] {$e_4$};
    \draw (n4) -- (n3) node [midway, fill=white, text opacity=1, opacity=0, right=2pt] {$e_5$};
    \draw (n4) -- (n5) node [midway, fill=white, text opacity=1, opacity=0, right=2pt] {$e_6$};
    \draw (n3) -- (n5) node [midway, fill=white, text opacity=1, opacity=0, below=2pt] {$e_7$};
    \draw (n2) -- (n6) node [midway, fill=white, text opacity=1, opacity=0, left=2pt] {$e_8$};
    \draw (n4) -- (n6) node [midway, fill=white, text opacity=1, opacity=0, right=2pt] {$e_9$};
    \end{tikzpicture}
    \caption{The graph used to define a $\{0, 1\}\SUB$ function in \Cref{prop:onsub-oxs}.}
    \label{fig:graph}
\end{figure}
\begin{prop}\label{prop:onsub-oxs}
    $\{0, 1\}\ONSUB \not\subset \text{OXS}$ 
\end{prop}
\begin{proof}
From \Cref{prop:a2-orderneutral} any $\{0,1\}\SUB$ is order-neutral. Consider a $\{0, 1\}\SUB$ function over $9$ elements $\{e_1, \dots, e_9\}$, each corresponding to an edge in the graph $G$ from Figure \ref{fig:graph} (i.e., a 6-node triangular lattice). Define $v(S)$ as the size of the largest subset of acyclic edges in $S$, thus $v$ is the rank function of the graphical matroid defined on $G$, and is therefore, $\{0, 1\}\SUB$ \citep{oxley2011matroids}.  

Suppose for contradiction that $v$ can be expressed as an OXS function $\hat v$ using unit demand valuations $f_1, \dots, f_k$. We divide this proof into a series of claims.

\begin{claim}\label{claim:1}
For all $j \in [k]$ and all $i \in [9]$, $f_j(e_i) \le 1$. Moreover, there is at least one unit demand function $f_j$ for each $i$ such that $f_j(e_i) = 1$.
\end{claim}
\begin{proof}
This stems from the fact that $v(e_i) = 1$ for all $i \in [9]$.
\end{proof}

For each $i \in [9]$, let $\cal F_i$ denote the set of functions $f_j$ such that $f_j(e_i) = 1$.

\begin{claim}\label{claim:2}
For any two edges $e_i$ and $e_{i'}$, we must have $|\cal F_i \cup \cal F_{i'}| \ge 2$.
\end{claim}
\begin{proof}
Since the shortest cycle in the graph has length three, $v(\{e_i, e_{i'}\}) = 2$. From \Cref{claim:1}, all the unit demand functions are upper bounded at $1$. Therefore, to obtain $\hat{v}(\{e_i, e_{i'}\}) = 2$, we must be able to partition $\{e_i, e_{i'}\}$ into two singleton sets such that each set can be given to a unit demand function that values the item in the set at $1$.
\end{proof}

\begin{claim}\label{claim:3}
For any $3$-cycle $\{e_i, e_{i'}, e_{i''}\}$ in the graph $G$, we must have $|\cal F_i \cup \cal F_{i'} \cup \cal F_{i''}| = 2$.
\end{claim}
\begin{proof}
We assume without loss of generality, the three cycle being considered is the set $S = \{e_1, e_2, e_3\}$. We know from \Cref{claim:2} that $|\cal F_1 \cup \cal F_{2} \cup \cal F_{3}| \ge 2$. Assume for contradiction that this inequality is strict and $|\cal F_1 \cup \cal F_{2} \cup \cal F_{3}| \ge 3$.  

Construct a bipartite graph $G' = (L \cup R, E)$ where $L = \{e_1, e_2, e_3\}$ and $R = \cal F_1 \cup \cal F_2 \cup \cal F_3$. There exists an edge from each $e_i$ to each function $f_j$ in $\cal F_i$. Note that, if there is a matching of size $3$ in this graph (also called an $L$-perfect matching), then $\hat v(S) = 3$ and we have a contradiction.

To show that there is an $L$-perfect matching in $G'$ when $|\cal F_1 \cup \cal F_{2} \cup \cal F_{3}| \ge 3$, we invoke Hall's theorem \citep{hall1935og}. Hall's theorem states that there is an $L$-perfect matching in $G'$ if for all subsets $W$ of $L$, the number of neighbors\footnote{A node $j$ is a neighbor of the set $W$ if and only if it has an edge to at least one node in $W$.} of $W$ in $G'$ has a size weakly greater than $|W|$. \Cref{claim:1} and \Cref{claim:2} show that this condition is satisfied when $|W| = 1$ and $|W| = 2$ respectively. If $|\cal F_1 \cup \cal F_{2} \cup \cal F_{3}| \ge 3$, then it satisfied for $|W| = 3$ as well and there is an $L$-perfect matching --- a contradiction. 

Therefore, we must have $|\cal F_1 \cup \cal F_{2} \cup \cal F_{3}| = 2$.
\end{proof}

Consider the sets $S = \{e_1, e_2, e_3\}$ and $S' = \{e_3, e_4, e_5\}$. From \Cref{claim:3}, we have $|\cal F_1 \cup \cal F_2 \cup \cal F_3| = 2$ and $|\cal F_3 \cup \cal F_4 \cup \cal F_5| = 2$. If $|\cal F_3| = 2$, it implies that $\cal F_4 \cup \cal F_5 = \cal F_3 = \cal F_1 \cup \cal F_2$ (using \Cref{claim:2}). This in turn implies that $\hat{v}(\{e_1, e_2, e_4\}) < 3$ which is a contradiction. Therefore, $|\cal F_3| = 1$.

We can use a similar argument to show that $|\cal F_4| = 1$ and $|\cal F_5| = 1$. From Claim \ref{claim:2}, we get that $\cal F_3$, $\cal F_4$ and $\cal F_5$ must be pairwise disjoint. This implies that $\hat{v}(\{e_3, e_4, e_5\}) = 3$ --- another contradiction since $\{e_3, e_4, e_5\}$ is a cycle. 

So we can conclude that the OXS function $\hat{v}$ does not exist. This completes the proof.
\end{proof}
Interestingly, \Cref{prop:onsub-oxs} also implies that the results of \citet{benabbou2021MRF} (who study $\{0, 1\}\OXS$ valuations) are strictly weaker than the results of \citet{Babaioff2021Dichotomous} (who study $\{0, 1\}\SUB$ valuations). To the best of our knowledge, this observation has not been formally made in the literature. 

There are two other popular classes that fall between additive valuations and submodular valuations --- Rado valuations \citep{garg2021rado} and Gross Subsitutes \citep{kelso1982gs}. We conjecture that both these classes strictly contain order-neutral submodular valuations but are unable to provide a proof. We leave this question open for future work.

\section{Conclusions and Future Work}
In this work, we study the computation of leximin allocations in instances with mixed goods and chores. We show that when agents have $\{-1, 0, c\}\ONSUB$ valuations, leximin allocations can be computed efficiently. We study these leximin allocations in further detail showing they are always Lorenz dominating and approximately proportional. 

On a higher level, our work is the first to generalize the path augmentation technique to tri-valued valuation functions. We are hopeful that the tools of weighted exchange graphs and decompositions can be applied to even more general valuation classes. We are also excited by the class of order-neutral submodular valuations. Much like Rado and OXS valuations, we believe order-neutral submodular valuations are an appealing sub-class of submodular valuation functions that warrant further study.
\bibliographystyle{plainnat}
\bibliography{abb,literature}

\end{document}